\documentclass[11pt]{article}
	
\makeatletter
\renewcommand\section{\@startsection {section}{1}{\z@}%
{-3.5ex \@plus -1ex \@minus -.2ex}%
{2.3ex \@plus.2ex}%
{\normalfont\fontfamily{phv}\fontsize{16}{19}\bfseries}}
\renewcommand\subsection{\@startsection{subsection}{2}{\z@}%
{-3.25ex\@plus -1ex \@minus -.2ex}%
{1.5ex \@plus .2ex}%
{\normalfont\fontfamily{phv}\fontsize{14}{17}\bfseries}}
\renewcommand\subsubsection{\@startsection{subsubsection}{3}{\z@}%
{-3.25ex\@plus -1ex \@minus -.2ex}%
{1.5ex \@plus .2ex}%
{\normalfont\normalsize\fontfamily{phv}\fontsize{14}{17}\selectfont}}
\makeatother
\usepackage{amsmath}
\usepackage{graphicx}
\usepackage{enumerate}
\usepackage{natbib} %comment out if you do not have the package
\usepackage{url} % not crucial - just used below for the URL
\usepackage{xcolor}
\usepackage{multirow}%
\usepackage{amssymb,amsfonts}%
\usepackage{amsthm}%
\usepackage{mathrsfs}%
\usepackage[title]{appendix}%
\usepackage{textcomp}%
\usepackage{manyfoot}%
\usepackage{booktabs}%
\usepackage{algorithm}%
\usepackage{algpseudocode}%
\usepackage{listings}%
\usepackage{float}
\usepackage{dblfloatfix}%
\usepackage{rotating}
\usepackage{longtable}
\usepackage{threeparttable}
\usepackage{dblfloatfix}
\usepackage{array}
\usepackage{makecell} 

\usepackage{hyperref}
\hypersetup{
    colorlinks=true,
    linkcolor=blue,
    filecolor=black,
    urlcolor=blue,
    citecolor=blue,
    anchorcolor=black,
    linktoc=all,
    breaklinks=true,
    pdfauthor={Jinho Cha},
    pdftitle={Dynamic Inverse Optimization under Drift and Shocks: Theory, Regret Bounds, and Applications}
}

\newtheorem{assumption}{Assumption}
\newtheorem{lemma}{Lemma}
\newtheorem{corollary}{Corollary}
\newcommand{\Proj}{\mathrm{P}}
\newtheorem{theorem}{Theorem}%
\newtheorem{proposition}{Proposition}%
\begin{document}
\sloppy		
%%%%%%%%%%%%%%%%%%%%%%%%%%%%%%%%%%%%%%%%%%%%%%%%%%%%%%%%%%%%%%%%%%%%%%%%%%%%%%
\def\spacingset#1{\renewcommand{\baselinestretch}%
	{#1}\small\normalsize} \spacingset{1}
%%%%%%%%%%%%%%%%%%%%%%%%%%%%%%%%%%%%%%%%%%%%%%%%%%%%%%%%%%%%%%%%%%%%%%%%%%%%%%
		
	\begin{center}
    {\LARGE\bf Dynamic Inverse Optimization under Drift and Shocks: Theory, Regret Bounds, and Applications}\\[0.85em]
    {\large Jinho Cha$^{*}$}\\[0.55em]
    {\normalsize Department of Computer Science, Gwinnett Technical College,\\
    5150 Sugarloaf Parkway, Lawrenceville, GA 30043, USA}\\[0.45em]
    {\normalsize Email: \href{mailto:jinhocha7@gmail.com}{jinhocha7@gmail.com}}\\[0.25em]
    {\normalsize ORCID: \href{https://orcid.org/0000-0001-9534-8090}{0000-0001-9534-8090}}\\[0.45em]
    {\small $^{*}$Corresponding author}
\end{center}
\medskip
\bigskip

\begin{abstract}
Operational decisions often reflect priorities that evolve gradually and may shift abruptly after disruptions. We develop a dynamic inverse optimization framework for recovering time-varying preference parameters from observed allocation trajectories. The framework separates an offline temporally regularized trajectory formulation from an online mirror-descent estimator for streaming recovery, thereby linking inverse optimization with online convex optimization under nonstationarity. We establish sufficient conditions for local identifiability and existence, static regret of $O(\sqrt{T})$, dynamic regret of $O(\sqrt{T}(1+V_T))$ for the online estimator, and a residual-to-parameter stability bound governed by the conditioning of the projected inverse map. The analysis distinguishes uniqueness of the inverse explanation, conditioning of the recovery map, existence of a batch trajectory solution, and sequential performance under a moving comparator. Controlled experiments on two stylized allocation benchmarks examine noiseless recoverability, finite-horizon regret, gradual drift, observational noise, abrupt shocks, and temporal coupling. These experiments are interpreted as finite-sample diagnostics rather than proofs of asymptotic bounds. Across the tested settings, temporal coupling moderates the trade-off between static under-adaptation and excessive variation from independent period-by-period inversion. The resulting framework provides an interpretable foundation for online inverse optimization and nonstationary decision analytics under evolving priorities. It also clarifies when structural recovery and sequential adaptation require distinct guarantees.\end{abstract}

% \noindent%
% {\it Keywords:} Dynamic Inverse Optimization, Inverse Optimization, Online Convex Optimization,
% Dynamic Regret, Identifiability, Temporal Regularization, Nonstationary Decision-Making,
% Decision Analytics.
\noindent%
		
%\newpage
\spacingset{1.5} % This is the required spacing for the manuscript body

\section{Introduction}
\label{s:intro}

\subsection{Motivation}

Many operational systems repeatedly allocate scarce resources even though the priorities behind those allocations are not directly observable. Hospitals assign capacity, power systems dispatch generation, logistics networks prioritize resources, and financial institutions rebalance portfolios. The decisions are visible; the trade-offs among cost, risk, fairness, resilience, and service quality are often implicit.

Classical optimization takes those objectives as known inputs. In practice, priorities can evolve with demand, regulation, resource availability, and disruptions. Some changes are gradual, whereas others occur abruptly. The resulting inverse problem is therefore not only to explain a single observed decision, but to recover and interpret an evolving sequence of preference parameters.

Two related questions motivate the analysis. First, under what structural conditions can latent preference parameters be recovered from observed allocation decisions? Second, when those parameters change over time, how can they be tracked sequentially without treating each observation as an unrelated static inverse problem?

The distinction matters. Local identifiability concerns uniqueness of the inverse explanation in an idealized neighborhood. It does not by itself control sensitivity to noise, numerical error, or temporal roughness. An inverse map may be locally one-to-one and still be poorly conditioned. Independent period-by-period recovery can therefore produce a highly variable estimated path even when each individual inverse problem is locally identifiable.

\subsection{Research Gap}

Inverse optimization (IO) offers a principled way to recover unknown objectives from observed decisions \citep{ahuja2001inverse,chan2014generalized}. Applications span transportation, energy systems, and healthcare \citep{bertsimas2015equilibrium,eshghi2016inverse}, and recent surveys document the breadth of the field \citep{chan2019closedform,chan2023theory}. Online inverse optimization has also begun to provide regret-based guarantees under sequential feedback \citep{sakaue2018online,dong2020generalized,sakaue2025online}.

Online convex optimization (OCO), in parallel, studies sequential learning under changing loss functions and supplies static and dynamic regret measures for nonstationary environments \citep{zinkevich2003online,besbes2015nonstationary,hall2015online}. Dynamic regret is commonly expressed in terms of a path-length or variation budget that quantifies movement of the comparator sequence.

These literatures address related but distinct questions. IO focuses on recovering latent objectives from observed decisions, whereas OCO focuses on sequential performance under evolving losses. Existing dynamic inverse approaches do not yet combine, in one allocation-based formulation, a clear separation among structural identifiability, temporal regularization, sequential updating, and regret relative to fixed and moving comparators.

A second gap concerns the role of sequential estimation once local identifiability is available. Local injectivity can establish uniqueness in a noiseless neighborhood, but it does not imply that period-by-period recovery is numerically stable or temporally well behaved. The motivation for OCO in this paper is therefore not to make an otherwise identifiable parameter identifiable. Its role is to supply a streaming update and a regret-based measure of adaptation when the target evolves and observations are perturbed. In this sense, the study positions dynamic inverse optimization as a bridge between inverse optimization and online convex optimization for nonstationary decision-making, while keeping identifiability, conditioning, and regret as distinct analytical objects.

\subsection{Contributions}

The framework separates structural recovery from sequential tracking. Its contributions are:

\begin{enumerate}
    \item \textbf{Dynamic inverse formulation with distinct structural and online losses.}
    A forward allocation model defines the behavioral structure, a Karush--Kuhn--Tucker-based structural diagnostic characterizes exact consistency in the noiseless benchmark, and a projected inverse loss supplies the signal used by the sequential estimator. A temporally regularized trajectory problem is introduced as an offline reference formulation rather than as an objective that the online recursion is claimed to solve exactly.

    \item \textbf{Identifiability, existence, and local stability.}
    Sufficient conditions are derived for local identifiability through the projected-gradient map, together with rank-based conditions for the linear-gradient case. Existence results are established for the trajectory formulation, and the conditioning of the projected Jacobian is used to quantify local residual-to-parameter sensitivity. This distinction separates uniqueness from stability.

    \item \textbf{Sequential estimator with static and dynamic regret guarantees.}
    An online mirror-descent recursion updates the latent parameter as observations arrive. Under the stated convexity, boundedness, and drift conditions, the recursion attains static regret of $O(\sqrt{T})$ and dynamic regret of $O(\sqrt{T}(1+V_T))$. These are guarantees for the sequential recursion; they do not assert exact equivalence with the offline regularized trajectory estimator.

    \item \textbf{Controlled computational diagnostics.}
    Two stylized allocation benchmarks are used to examine noiseless recoverability, finite-horizon regret behavior, drift sensitivity, perturbation effects, abrupt shocks, and temporal coupling. The experiments are interpreted as diagnostics of mechanisms emphasized by the analysis, not as proofs of the theoretical rates. In particular, the estimator is evaluated against a static inverse baseline and independent period-by-period inversion to expose the trade-off between under-adaptation and excessive temporal variation.
\end{enumerate}

Together, these contributions connect inverse geometry and online tracking while keeping identifiability, conditioning, and temporal variation conceptually distinct.

\subsection{Paper Structure}

The remainder of the paper is organized as follows.
Section~\ref{sec:literature} reviews the related literature and positions the study within inverse optimization and nonstationary learning.
Section~\ref{sec:model} introduces the forward allocation model and formulates the dynamic inverse problem.
Section~\ref{sec:theory} presents the online estimator and theoretical results on identifiability, existence, stability, and regret.
Section~\ref{sec:experiments} reports controlled computational experiments.
Section~\ref{sec:implications} discusses methodological and operational implications.
Section~\ref{sec:conclusion} concludes the paper.

%%%%%%%%%%%%%%%%%%%%%%%%%%%%%%
\section{Related Literature}
\label{sec:literature}

\subsection{Forward Allocation Models}

Forward allocation problems remain central in operations research, with recent
work emphasizing robustness, fairness, and learning. Robust optimization
offers tractable methods under uncertainty. Foundational contributions by
\citet{bental2002robust} established rigorous robust counterparts, while
\citet{bertsimas2004robust} introduced the influential price of robustness,
shaping applications in supply chains, healthcare, and energy systems.
Fairness has also become an important theme in resource allocation.
\citet{bertsimas2011price} analyzed the price of fairness, and
\citet{caragiannis2012efficiency} examined efficiency--fairness
trade-offs in allocation mechanisms. Learning-based approaches further
expanded stochastic allocation models, with \citet{levi2007provably} and
\citet{besbes2013ordering} showing how adaptive, data-driven policies can
improve performance under uncertainty.

Despite this progress, most forward optimization frameworks assume that
objectives are explicitly specified a priori rather than inferred from
observed decisions. This distinction motivates the inverse perspective
studied in this paper.

\subsection{Inverse Optimization Approaches}

IO seeks to recover latent objectives from observed decisions. Foundational work established tractable inverse formulations and demonstrated applications in network optimization and healthcare planning \citep{ahuja2001inverse,chan2014generalized,bertsimas2015equilibrium}. More recent surveys document the rapid methodological and application-oriented expansion of the field \citep{chan2019closedform,chan2023theory}.

Recent studies have begun to relax stationarity assumptions. Time-varying inverse optimal control has been examined using sliding-window estimation \citep{westermann2020timevarying}, while neural-network-based inverse optimal control has been developed for nonlinear vehicle systems \citep{cespi2023neural}. In transportation settings, \citet{xu2024schedule} infer commuters' schedule preferences and crowding perceptions from aggregate smart-card observations. Online inverse optimization has also emerged as a promising direction, with regret-based learning guarantees under sequential feedback \citep{sakaue2018online,dong2020generalized,sakaue2025online}.

Nevertheless, many IO formulations remain static, while dynamic extensions often rely on local approximations or restrictive preference structures. A formulation that combines local inverse structure with sequential tracking under drift is still comparatively limited.

\subsection{Dynamic Optimization under Drift}

OCO and nonstationary learning study sequential
decision-making when losses or environments evolve over time. Central to this
literature are variation budgets and path length, which quantify cumulative
change and support dynamic regret guarantees.
\citet{zinkevich2003online} introduced the OCO framework, and
\citet{besbes2015nonstationary} established sharp regret bounds in nonstationary
stochastic optimization under variation budgets.
\citet{hall2015online} demonstrated broad applicability through path-length
analysis in signal processing and statistical learning.

Subsequent work refined these guarantees. \citet{chiang2012online} developed adaptive bounds responsive to gradual variation, \citet{mokhtari2016online} proposed improved rates in strongly convex settings, and \citet{zhao2021dynamic} derived tighter variation-sensitive bounds for strongly convex and smooth losses. Together, these results establish regret as a useful measure of adaptation in drifting environments.

However, these studies address forward loss minimization rather than the
recovery of latent objectives from observed behavior. Their application to latent-objective recovery under drift remains comparatively underdeveloped.

\subsection{Adjacent Areas: Inverse Optimal Control, Inverse Reinforcement Learning, and Dynamic Inverse Problems}

In control and reinforcement learning, inverse optimal control (IOC) and inverse reinforcement learning (IRL) recover cost or reward functions from observed trajectories. Classical formulations include \citep{todorov2005generalized,abbeel2004apprenticeship}. More recent work has extended batch inverse reinforcement learning to settings with multiple latent intentions \citep{ramponi2020batch}. These approaches highlight preference recovery from demonstrations, but their guarantees typically focus on policy learning, stability, or sample complexity rather than identifiability and regret in allocation settings.

In applied mathematics, dynamic inverse problems use variational, low-rank, measure-valued, and optimal-transport formulations to reconstruct time-dependent unknowns. Representative contributions include \citet{bredies2020optimaltransport}, \citet{arridge2022joint}, and \citet{bredies2023gcg}. Although mathematically rigorous, these methods target different forward models and do not address behavioral preference recovery in operational allocation systems.

\subsection{Positioning of This Paper}

Three themes are especially relevant to the present study. First, inverse optimization has
largely focused on static objectives, while online variants mainly analyze
regret under fixed or linear preference structures. Second, OCO under drift
yields static and dynamic regret guarantees for forward optimization, but
does not recover latent objectives. Third, adjacent areas such as IOC, IRL,
and dynamic inverse problems incorporate time variation, yet emphasize
stability, policy recovery, or operator estimation rather than
identifiability and regret for allocation decisions.

The intersection studied here---recovering time-varying preferences in
allocation models with explicit identifiability conditions, temporal
regularization, and static/dynamic regret guarantees---has received limited
systematic attention. The present formulation links inverse optimization with nonstationary learning while retaining structural identifiability, temporal regularization, and sequential regret as distinct analytical components.

\section{Model and Problem Formulation}
\label{sec:model}

\subsection{Forward Allocation Problem}
\label{subsec:forward}

Consider a sequence of allocation problems indexed by $t = 1,\dots,T$. At each
period, the observed allocation is assumed to arise from the solution of a parametric
forward problem

\begin{equation}
\label{eq:forward}
\min_{\mathbf{x} \in \mathcal{X}_t} \; c_t(\mathbf{x}; \boldsymbol{\theta}_t),
\end{equation}

where $\mathcal{X}_t$ denotes the feasible set and
$c_t(\cdot;\boldsymbol{\theta}_t)$ is a cost function parameterized by a latent
preference vector $\boldsymbol{\theta}_t \in \Theta \subset \mathbb{R}^p$.

The parameter $\boldsymbol{\theta}_t$ represents the unobserved trade-off weights
that govern the decision-maker's behavior at time $t$. In different applications,
it may encode, for example, priority weights in healthcare or cost--emission
trade-offs in energy dispatch. The analyst does not observe
$\boldsymbol{\theta}_t$ directly, but only the resulting allocation decisions.

\subsubsection{Decision variables and constraints}

At each period $t$, a decision-maker allocates resources
$\mathbf{x}_t \in \mathbb{R}^n_{+}$ across $n$ agents. The feasible set is

\begin{equation}
\label{eq:feasible-set}
\mathcal{X}_t :=
\left\{
\mathbf{x} \in \mathbb{R}^n_{+}
\;\middle|\;
\mathbf{B}_t \mathbf{x} \le \mathbf{q}_t
\right\},
\end{equation}

where $\mathbf{q}_t \in \mathbb{R}^k_{+}$ is the vector of available capacities and
$\mathbf{B}_t \in \mathbb{R}^{k \times n}$ is the resource consumption matrix.

This representation accommodates both single-resource and multi-resource allocation
settings. When $k=1$ and $\mathbf{B}_t = \mathbf{1}^\top$, the formulation reduces
to a standard aggregate-capacity constraint. For $k>1$, it captures coupled capacity
limits, such as intensive care unit (ICU) beds and general beds in healthcare or generation and fuel
constraints in energy systems.

\subsubsection{Cost functions}

The cost function
\[
c_t : \mathbb{R}^n \times \Theta \to \mathbb{R},
\qquad
c_t(\mathbf{x}; \boldsymbol{\theta}_t),
\]
captures the latent trade-offs underlying the observed allocation. The structural role
of $\boldsymbol{\theta}_t$ is deliberately kept general so that the formulation covers
both linear and nonlinear forward models.

\subsection{Inverse Problem Definition}
\label{subsec:inverse}

The forward model maps latent preferences $\boldsymbol{\theta}_t$ to allocations
$\mathbf{x}_t$. The inverse problem reverses this relationship: given observations
of allocations and constraints, the goal is to infer the latent preference sequence
$\{\boldsymbol{\theta}_t\}_{t=1}^T$ that best rationalizes the observed decisions.

Although the forward problem in Equation~\eqref{eq:forward} is standard, the inverse problem
is structurally more difficult. First, the analyst does not observe the dual variables
associated with the constraints. Second, the inverse mapping is typically defined only
implicitly through optimality conditions rather than through a closed-form expression.
Third, even when the inverse mapping is locally unique, it may be ill-conditioned in
the presence of noise, so that small perturbations in observed decisions can induce
large variation in recovered parameters. These features make the inverse problem
substantially different from solving the forward model period by period.

\subsubsection{Karush--Kuhn--Tucker-Based Structural Loss}

Let the Karush--Kuhn--Tucker (KKT) conditions for Equation~\eqref{eq:forward} be
written first for the capacity constraints. For notational economy, the displayed system
assumes that the nonnegativity bounds are inactive at the observed decision. If a component
of $\mathbf{x}_t$ is at zero, its active bound normal $-\mathbf{e}_j^\top\mathbf{x}\le 0$
can be appended to the active-constraint matrix and the same analysis applied to the augmented
normal system.
\begin{equation}
\label{eq:KKT}
\nabla_{\mathbf{x}} c_t(\mathbf{x}_t; \boldsymbol{\theta})
+ \mathbf{B}_t^\top \boldsymbol{\lambda}_t = \mathbf{0},
\qquad
\boldsymbol{\lambda}_t \ge \mathbf{0},
\qquad
\mathbf{B}_t \mathbf{x}_t \le \mathbf{q}_t,
\qquad
\boldsymbol{\lambda}_t^\top(\mathbf{B}_t \mathbf{x}_t - \mathbf{q}_t)=0,
\end{equation}
where $\boldsymbol{\lambda}_t \in \mathbb{R}^k_{+}$ denotes the vector of Lagrange
multipliers.

\begin{assumption}[Structural consistency]
\label{ass:KKT}
For each $t$, the observed decision $\mathbf{x}_t$ is consistent with some
underlying forward-optimal solution in the noiseless idealization.
\end{assumption}

Assumption~\ref{ass:KKT} establishes the structural baseline for identifiability analysis.
It is used only to characterize the noiseless benchmark in which observed allocations
are exactly rationalizable by the forward model. Later sections relax this idealization
by allowing perturbations due to observational noise and approximation error.

The structural KKT diagnostic is defined as the residual violation of the KKT system:
\begin{equation}
\label{eq:inverse-loss}
\mathcal{L}^{\mathrm{KKT}}_t(\boldsymbol{\theta})
:=
\inf_{\boldsymbol{\lambda}_t \ge \mathbf{0}}
\Bigg(
\|(\mathbf{B}_t \mathbf{x}_t - \mathbf{q}_t)_+\|^2
+
\|\nabla_{\mathbf{x}} c_t(\mathbf{x}_t; \boldsymbol{\theta})
+ \mathbf{B}_t^\top \boldsymbol{\lambda}_t\|^2
+
\Big|
\boldsymbol{\lambda}_t^\top(\mathbf{B}_t \mathbf{x}_t - \mathbf{q}_t)
\Big|
\Bigg).
\end{equation}

The three terms correspond, respectively, to violations of primal feasibility,
stationarity, and complementarity. The diagnostic therefore measures how well a candidate
parameter $\boldsymbol{\theta}$ explains the observed decision through the optimality
system of the forward model.

This formulation is structural rather than statistical. It does not assume a likelihood
model for the observations. Instead, it evaluates whether the observed decision is
approximately compatible with optimal behavior under the candidate parameter. This
choice is important because, in many operational settings, the analyst can write down
the forward model and the constraint system, but cannot credibly specify a full
probabilistic model for the latent preferences.

Even when identifiability holds locally, direct per-period inversion based on the KKT
system remains problematic. Recovering $\boldsymbol{\theta}_t$ at a single period
requires solving simultaneously for both the unknown preference parameter and the
unobserved multiplier vector $\boldsymbol{\lambda}_t$. In general, this produces an
implicit nonlinear system whose numerical solution may be unstable under noise and
whose period-by-period estimates may exhibit large variance. For this reason, the
inverse problem should not be viewed as a collection of independent static inversions.

\subsubsection{Trajectory-based recovery}

Because the latent preference sequence evolves over time, recovery is formulated as a
trajectory estimation problem rather than as a sequence of isolated inversions:
\begin{equation}
\label{eq:trajectory-estimator}
\{\hat{\boldsymbol{\theta}}_t\}_{t=1}^T
\in
\arg\min_{\boldsymbol{\theta}_1,\dots,\boldsymbol{\theta}_T \in \Theta}
\left\{
\sum_{t=1}^T \mathcal{L}^{\mathrm{KKT}}_t(\boldsymbol{\theta}_t)
+
\lambda \sum_{t=2}^T d(\boldsymbol{\theta}_t,\boldsymbol{\theta}_{t-1})
\right\},
\end{equation}
where $\lambda > 0$ is a regularization parameter and $d(\cdot,\cdot)$ is a metric
or divergence on $\Theta$, such as Euclidean distance, $\ell_1$ distance, or a
Bregman divergence.

The first term enforces period-by-period consistency with the observed decisions.
The second term penalizes excessive temporal variation and couples the parameters
across time. This coupling directly regularizes temporal variation in dynamic environments. When observations
are noisy, separate inversion at each period can amplify perturbations and produce
erratic estimates. By contrast, the trajectory formulation borrows strength across
neighboring periods and stabilizes recovery through temporal regularity.

The trajectory formulation also has a substantive interpretation. In many applications,
the latent parameter is not expected to change arbitrarily from one period to the next.
Operational priorities drift because of gradual shifts in demand, regulation, congestion,
or market conditions, and only occasionally exhibit abrupt breaks. The regularization
term in Equation~\eqref{eq:trajectory-estimator} encodes precisely this feature by favoring
preference paths that vary over time but do not oscillate without operational reason.

This formulation addresses two limitations of independent per-period inversion. First,
it reduces estimation instability by suppressing spurious fluctuations induced by noisy
observations. Second, it respects the dynamic nature of the problem by recognizing that
the object of interest is the entire parameter path, not a disconnected collection of
single-period estimates.

At the same time, Equation~\eqref{eq:trajectory-estimator} is a batch problem defined over the
entire horizon. As $T$ grows, direct solution becomes increasingly expensive, especially
in settings where data arrive sequentially. This motivates the recursive approximation
developed later through online convex optimization.

\subsection{Dynamic Drift}
\label{subsec:drift}

Temporal variation in the latent parameter sequence is quantified by the
variation budget
\begin{equation}
\label{eq:VT}
V_T
:=
\sum_{t=2}^T d(\boldsymbol{\theta}_t,\boldsymbol{\theta}_{t-1}),
\end{equation}
a standard complexity measure in nonstationary learning
\citep{hazan2009changing,besbes2015nonstationary}.

The variation budget plays two roles. Analytically, it measures the amount of drift in
the underlying preference trajectory and therefore determines the difficulty of tracking
the evolving parameter. From a modeling perspective, it is the natural counterpart to the
temporal regularization term in Equation~\eqref{eq:trajectory-estimator}. A slowly varying
parameter path corresponds to small $V_T$, whereas abrupt structural changes increase
$V_T$ and make recovery more difficult.

The central difficulty in dynamic inverse optimization is that the analyst faces both
cross-sectional ambiguity and temporal nonstationarity. Cross-sectionally, the parameter
must be recovered indirectly from decisions through the KKT structure. Temporally, the
target itself moves. A satisfactory method must therefore do more than identify
$\boldsymbol{\theta}_t$ locally in an idealized noiseless setting; it must also track the
parameter path as new observations arrive.

For this reason, re-solving the full batch estimator in Equation~\eqref{eq:trajectory-estimator}
after every new observation is generally unattractive. Such a strategy is computationally
burdensome, ignores the sequential structure of the problem, and offers no direct control
of tracking performance over time. What is needed instead is a recursive estimation rule
that updates the parameter as data arrive while preserving the stabilizing effect of the
trajectory formulation.

OCO supplies such a mechanism. In the present context,
OCO is not introduced merely as a generic algorithmic tool. Its role is structural.
First, it furnishes recursive updates without repeatedly solving the full-horizon trajectory problem. Second, its Bregman-proximal update discourages abrupt estimator movement while remaining distinct from the explicit path penalty in the offline formulation. Third, OCO yields
performance guarantees in terms of regret relative to an evolving benchmark, which gives
a rigorous notion of tracking quality in nonstationary environments.

The connection is therefore conceptual rather than algorithmic equivalence: the offline formulation makes temporal regularization explicit, whereas OCO supplies an implementable sequential estimator with its own regret analysis. The distinction is substantive. Identifiability establishes when the inverse problem is
locally well posed; the trajectory formulation stabilizes recovery over time; and OCO
supplies the recursion needed for such recovery in dynamic settings.

%--------------------------------------------------
\section{Online Estimation Algorithm and Theoretical Guarantees}
\label{sec:theory}

This section introduces the sequential estimator before analyzing its properties. Algorithm~\ref{alg:unified-strong} updates the latent parameter from each arriving allocation using a projected inverse loss; Theorems~\ref{thm:static-regret} and~\ref{thm:dynamic-regret} provide the corresponding static and dynamic regret guarantees. The algorithm is related to, but distinct from, the offline temporally regularized trajectory formulation in Section~\ref{sec:model}: the latter provides a batch reference problem whose existence is established in Theorem~\ref{thm:exist-unique}, whereas the former is the streaming procedure analyzed by regret.

The assumptions are correspondingly modular. Assumptions~\ref{ass:KKT}--\ref{ass:inj} govern structural recoverability; Assumptions~\ref{ass:level}--\ref{ass:theta-compact} guarantee well-posedness; Assumptions~\ref{ass:subgrad}--\ref{ass:drift} support the sequential regret analysis; and Assumption~\ref{ass:noise} controls local perturbation sensitivity. This separation is important because identifiability, conditioning, existence, and online tracking are logically distinct properties and should not be inferred from one another.

The analysis proceeds in the order needed to interpret these two objects. Section~\ref{subsec:theory-overview} summarizes their relationship. Section~\ref{subsec:algorithm} states the online estimator. Section~\ref{subsec:identifiability} establishes the structural conditions under which the projected inverse signal identifies the latent parameter and clarifies the distinction between uniqueness and conditioning. Section~\ref{subsec:existence} proves existence and well-posedness for the forward problem and the offline trajectory formulation. Section~\ref{subsec:regret} analyzes the sequential estimator against fixed and moving comparators, and Section~\ref{subsec:noise} converts projected residual error into a local parameter-error bound. Each result states its required assumptions locally; complete proofs are collected in Online Supplement Sections~S1.2--S1.14.

For reference, Table~\ref{tab:notation} collects the notation used repeatedly in the analysis.

\begin{table}[!htbp]
\caption{Summary of notation used in the theoretical analysis}
\label{tab:notation}
\centering
\spacingset{1}
\renewcommand{\arraystretch}{1.0}
\setlength{\extrarowheight}{0pt}
\begin{tabular}{p{0.36\textwidth}p{0.58\textwidth}}
\toprule
\textbf{Symbol} & \textbf{Description} \\
\midrule

\multicolumn{2}{l}{\textbf{Decision variables}} \\
\hspace{1.5em}$\mathbf{x}_t \in \mathbb{R}^n_{+}$ & Allocation decision vector at period $t$ \\
\hspace{1.5em}$\hat{\boldsymbol{\theta}}_t$ & Estimated latent parameter at period $t$ \\
\hspace{1.5em}$\hat{\boldsymbol{\theta}}_{1:T}$ & Estimated parameter trajectory $\{\hat{\boldsymbol{\theta}}_t\}_{t=1}^T$ \\

\multicolumn{2}{l}{\textbf{Parameters and primitives}} \\
\hspace{1.5em}$\boldsymbol{\theta}_t \in \Theta \subset \mathbb{R}^p$ & Latent parameter at period $t$ \\
\hspace{1.5em}$c_t(\mathbf{x};\boldsymbol{\theta})$ & Cost function at period $t$ \\
\hspace{1.5em}$\mathbf{B}_t \in \mathbb{R}^{k \times n}$ & Constraint matrix at period $t$ \\
\hspace{1.5em}$\mathbf{q}_t \in \mathbb{R}^k$ & Right-hand side vector at period $t$ \\
\hspace{1.5em}$\mathbf{A}_t(\mathbf{x}) \in \mathbb{R}^{n \times p}$ & Jacobian of $\nabla_{\mathbf{x}} c_t(\mathbf{x};\boldsymbol{\theta})$ w.r.t.\ $\boldsymbol{\theta}$ \\
\hspace{1.5em}$\mathbf{b}_t(\mathbf{x}) \in \mathbb{R}^n$ & Offset in affine gradient representation \\
\hspace{1.5em}$\Theta$ & Feasible parameter set \\

\multicolumn{2}{l}{\textbf{Dual variables}} \\
\hspace{1.5em}$\boldsymbol{\lambda}_t \in \mathbb{R}^k_{+}$ & Lagrange multipliers at period $t$ \\

\multicolumn{2}{l}{\textbf{Operators and functions}} \\
\hspace{1.5em}$\mathcal{L}^{\mathrm{KKT}}_t(\boldsymbol{\theta})$ & Structural KKT diagnostic used for noiseless consistency and offline trajectory formulation \\
\hspace{1.5em}$\ell_t(\boldsymbol{\theta})$ & Projected online loss $\|\mathbf r_t(\boldsymbol{\theta})\|^2$ used by Algorithm~\ref{alg:unified-strong} and regret analysis \\
\hspace{1.5em}$\Proj_t$ & Orthogonal projector onto $\ker(\mathbf{B}_t)$ defined by $\mathbf{I} - \mathbf{B}_t^\top (\mathbf{B}_t \mathbf{B}_t^\top)^\dagger \mathbf{B}_t$ \\
\hspace{1.5em}$\mathbf{I}$ & Identity matrix \\
\hspace{1.5em}$\mathbf{1}$ & All-ones vector \\
\hspace{1.5em}$\Pi_{\Theta}(\mathbf{v})$ & Euclidean projection of $\mathbf{v}$ onto $\Theta$ \\
\hspace{1.5em}$\langle \mathbf{u},\mathbf{v}\rangle$ & Euclidean inner product $\mathbf{u}^\top \mathbf{v}$ \\
\hspace{1.5em}$\|\mathbf{u}\|$ & Euclidean norm \\

\multicolumn{2}{l}{\textbf{Structural and statistical quantities}} \\
\hspace{1.5em}$V_T$ & Variation budget $\sum_{t=2}^T \|\boldsymbol{\theta}_t - \boldsymbol{\theta}_{t-1}\|$ \\
\hspace{1.5em}$G$ & Uniform subgradient bound \\
\hspace{1.5em}$\lambda$ & Temporal regularization weight in trajectory objective \\

\bottomrule
\end{tabular}
\end{table}
\spacingset{1.5}

\subsection{Overview}
\label{subsec:theory-overview}

Two inverse quantities serve different purposes. The structural diagnostic $\mathcal{L}^{\mathrm{KKT}}_t(\boldsymbol{\theta})$ characterizes exact KKT consistency and enters the offline trajectory formulation, whereas the projected loss $\ell_t(\boldsymbol{\theta})$ supplies the signal used by Algorithm~\ref{alg:unified-strong}. The distinction is deliberate: structural identifiability concerns whether an inverse explanation is locally unique, while sequential estimation concerns how an evolving target can be tracked as observations arrive.

Local uniqueness alone is not enough for stable recovery. When the projected inverse map is poorly conditioned, small residual perturbations can induce large parameter deviations, and independent period-by-period inversion can generate an excessively rough path. The offline formulation controls path roughness through an explicit temporal penalty, while the online estimator uses local Bregman geometry and is evaluated through regret against fixed and moving comparators. These mechanisms are related through the same projected inverse structure but are not algorithmically equivalent.

The results below therefore answer four complementary questions: when the latent parameter is locally identifiable (Section~\ref{subsec:identifiability}); whether the forward and offline inverse problems attain solutions (Section~\ref{subsec:existence}); how the online estimator performs under stationary and drifting comparators (Section~\ref{subsec:regret}); and how inverse conditioning controls sensitivity to residual perturbations (Section~\ref{subsec:noise}). This separation is also used in Section~\ref{sec:experiments} to interpret drift, noise, shocks, and temporal coupling without conflating structural recoverability with sequential tracking.

This ordering also clarifies how the individual results support the estimator. Lemma~\ref{lem:proj-properties} justifies the projector used in Algorithm~\ref{alg:unified-strong}; Lemmas~\ref{lem:kkt-equivalence} and~\ref{lem:proj-consistency} show why its projected residual carries inverse-optimality information. Theorem~\ref{thm:ident} and its stationary consequences identify when that signal determines the latent parameter locally, while Theorem~\ref{thm:linear-ident} and Corollary~\ref{cor:rank} make recoverability and conditioning verifiable through rank and singular values in the affine-gradient case. Lemma~\ref{lem:identifiability-gap} then explains why these pointwise properties do not by themselves stabilize an estimated path. Theorem~\ref{thm:exist-unique} guarantees that the related offline trajectory benchmark is well defined; Theorems~\ref{thm:static-regret}--\ref{thm:dynamic-regret} evaluate the online recursion itself; and Theorem~\ref{thm:noise} translates its projected residual into a local parameter-error bound. Thus, the section does not present a collection of independent technical statements: the results move from validity of the inverse signal, to uniqueness and conditioning, to existence of the batch reference problem, and finally to sequential performance and perturbation sensitivity.

\subsection{Algorithm}
\label{subsec:algorithm}

The offline trajectory estimator in Equation~\eqref{eq:trajectory-estimator} jointly recovers $\boldsymbol{\theta}_{1:T}$ by combining structural KKT fit with an explicit path penalty. In a streaming setting, however, future observations are unavailable and repeatedly resolving the full trajectory problem is unnecessary. Algorithm~\ref{alg:unified-strong} therefore uses a separate recursive update based on the projected residual
\[
\mathbf r_t(\boldsymbol{\theta})=\Proj_t\nabla_{\mathbf x}c_t(\tilde{\mathbf x}_t;\boldsymbol{\theta}),
\qquad
\ell_t(\boldsymbol{\theta})=\|\mathbf r_t(\boldsymbol{\theta})\|^2.
\]
The projector removes constraint-induced gradient components, so this loss is aligned with the identifiability analysis in Section~\ref{subsec:identifiability}. A mirror-descent step then produces the next estimate from current first-order information. The Bregman term regularizes each local update but is not the explicit path penalty of the offline objective.

\begin{algorithm}[!htbp]
\caption{Structure-Aware Online Inverse Optimization under Drift}
\label{alg:unified-strong}
\small
\begin{algorithmic}[1]
\Require Convex set $\Theta$, step sizes $\{\eta_t\}$, strongly convex mirror map $\psi$, observations $\{(B_t,q_t,\tilde{\mathbf x}_t)\}_{t=1}^T$
\Ensure Estimates $\{\hat{\boldsymbol{\theta}}_t\}_{t=1}^T$
\State Initialize $\hat{\boldsymbol{\theta}}_1\in\Theta$
\For{$t=1,\ldots,T$}
    \State $\Proj_t\gets I-B_t^\top(B_tB_t^\top)^\dagger B_t$
    \State $\mathbf r_t(\boldsymbol{\theta})\gets\Proj_t\nabla_{\mathbf x}c_t(\tilde{\mathbf x}_t;\boldsymbol{\theta})$, $\ell_t(\boldsymbol{\theta})\gets\|\mathbf r_t(\boldsymbol{\theta})\|^2$
    \State Choose $\mathbf g_t\in\partial\ell_t(\hat{\boldsymbol{\theta}}_t)$
    \If{$t<T$}
      \State $\displaystyle \hat{\boldsymbol{\theta}}_{t+1}\gets\arg\min_{\boldsymbol{\theta}\in\Theta}\left\{\langle\mathbf g_t,\boldsymbol{\theta}\rangle+\eta_t^{-1}D_\psi(\boldsymbol{\theta},\hat{\boldsymbol{\theta}}_t)\right\}$
    \EndIf
\EndFor
\State \Return $\{\hat{\boldsymbol{\theta}}_t\}_{t=1}^T$
\end{algorithmic}
\end{algorithm}

% \medskip

Algorithm~\ref{alg:unified-strong} is the sequential object analyzed below. In the Euclidean case $D_\psi(\theta,\theta')=\frac12\|\theta-\theta'\|^2$, the update reduces to projected online gradient descent. Theorems~\ref{thm:static-regret} and~\ref{thm:dynamic-regret} bound its cumulative inverse loss, while Theorem~\ref{thm:noise} links a small projected residual to local parameter accuracy. Practical safeguards such as clipping or adaptive step sizes may be useful numerically but are not part of the stated theoretical recursion.

The distinction between the offline and online objects can be stated more concretely. The offline estimator chooses an entire path $\boldsymbol{\theta}_{1:T}$ after the full sequence is available and balances structural KKT fit against an explicit intertemporal penalty. Algorithm~\ref{alg:unified-strong}, by contrast, receives one observation at a time and updates only through the current projected loss and the local Bregman term. Consequently, Theorem~\ref{thm:exist-unique} establishes attainment for the batch reference problem but is not a convergence theorem for Algorithm~\ref{alg:unified-strong}; conversely, Theorems~\ref{thm:static-regret}--\ref{thm:dynamic-regret} evaluate the online recursion without requiring it to reproduce the batch minimizer. The two formulations share the same structural inverse signal and both favor temporal regularity, but they answer different operational questions. The batch formulation asks what globally regular path best explains a completed sequence, whereas the online algorithm asks how well a streaming estimate performs while the target itself may move. This separation is maintained throughout the theory and experiments so that batch well-posedness, online regret, and local perturbation stability are not conflated.

\subsection{Identifiability}
\label{subsec:identifiability}

The identifiability analysis asks whether the latent parameter is uniquely determined by the inverse optimality signal in a noiseless neighborhood. Assumption~\ref{ass:KKT} provides structural KKT consistency, and the following condition supplies the required local injectivity.

\begin{assumption}[Local projected-gradient injectivity]
\label{ass:inj}
The mapping
\[
\boldsymbol{\theta}\mapsto\Proj_t\nabla_{\mathbf x}c_t(\mathbf x_t;\boldsymbol{\theta})
\]
is locally injective in a neighborhood of $\boldsymbol{\theta}_t$.
\end{assumption}

The results form a single logical chain. Lemma~\ref{lem:proj-properties} establishes the projection geometry; Lemmas~\ref{lem:kkt-equivalence}--\ref{lem:proj-consistency} connect zero structural loss to a parameter-only projected first-order condition by eliminating the dual term. Theorem~\ref{thm:ident} converts that condition into local parameter uniqueness. Corollary~\ref{cor:stationary} and Proposition~\ref{prop:consistency-implies-stationarity} clarify the stationary benchmark and distinguish common-parameter consistency from genuine dynamics. Theorem~\ref{thm:linear-ident} then replaces abstract injectivity with a verifiable rank condition in the affine-gradient case; Corollary~\ref{cor:rank} quantifies the associated conditioning strength; and Lemma~\ref{lem:identifiability-gap} explains why pointwise identifiability alone does not guarantee a stable recovered trajectory. Thus the sequence moves from geometry, to uniqueness, to conditioning, and finally to the motivation for temporal coupling.

The hierarchy is useful for interpreting Algorithm~\ref{alg:unified-strong}. The algorithm does not require an explicit estimate of the dual multipliers at every period; Lemmas~\ref{lem:proj-properties}--\ref{lem:proj-consistency} justify replacing that augmented KKT representation by the projected residual used in its loss. The identifiability results then state when a zero residual has a unique local parameter interpretation, while the rank and stability results quantify whether small nonzero residuals remain informative. This distinction is essential under noisy observations: a residual may be small yet correspond to a sizable parameter error when the projected Jacobian is poorly conditioned. Lemma~\ref{lem:identifiability-gap} carries this point from a single period to the temporal setting by showing why repeated pointwise inversion can yield an unnecessarily variable path.

\begin{lemma}[Projector properties]
\label{lem:proj-properties}
The operator $\Proj_t$ is the orthogonal projector onto 
$\ker(\mathbf{B}_t)$.
\end{lemma}

\noindent\textit{Proof sketch.} The result follows from the Moore--Penrose identities: $\Proj_t$ is symmetric and idempotent, fixes $\ker(B_t)$, and annihilates $\mathrm{range}(B_t^\top)$. Complete proof: Online Supplement Section~S1.2.

\begin{lemma}[KKT equivalence]
\label{lem:kkt-equivalence}
$\mathcal{L}^{\mathrm{KKT}}_t(\boldsymbol{\theta}) = 0$ if and only if $\mathbf{x}_t$ satisfies the
KKT conditions of the forward problem under parameter $\boldsymbol{\theta}$ for some admissible multiplier $\boldsymbol{\lambda}_t\ge 0$.
\end{lemma}

\noindent\textit{Proof sketch.} Zero structural diagnostic is equivalent to vanishing feasibility, stationarity, and complementarity residuals for an admissible multiplier. Complete proof: Online Supplement Section~S1.3.

\begin{lemma}[Projected-gradient consistency]
\label{lem:proj-consistency}
If $\mathcal{L}^{\mathrm{KKT}}_t(\boldsymbol{\theta}) = 0$, then
\[
\Proj_t \nabla_{\mathbf{x}} c_t(\mathbf{x}_t;\boldsymbol{\theta}) = 0.
\]
\end{lemma}

\noindent\textit{Proof sketch.} Projecting KKT stationarity eliminates $B_t^\top\lambda_t$ because $\Proj_tB_t^\top=0$. Complete proof: Online Supplement Section~S1.4.

These properties reduce the inverse problem to a projected-gradient equation 
that depends only on $\boldsymbol{\theta}$ and eliminates the need to 
explicitly track dual variables.

\begin{theorem}[Local identifiability]
\label{thm:ident}
Under Assumptions~\ref{ass:KKT} and~\ref{ass:inj}, if 
$\mathcal{L}^{\mathrm{KKT}}_t(\boldsymbol{\theta}) = 0$ and $\boldsymbol{\theta}$ lies in a 
neighborhood of $\boldsymbol{\theta}_t$, then
\[
\boldsymbol{\theta} = \boldsymbol{\theta}_t.
\]
\end{theorem}

\noindent\textit{Proof sketch.} Theorem~\ref{thm:ident} follows by applying Lemma~\ref{lem:proj-consistency} to both the candidate and true parameters and invoking local injectivity. Complete proof: Online Supplement Section~S1.5.

Theorem~\ref{thm:ident} is deliberately local: it establishes uniqueness near the true parameter but does not imply robust recovery. When the projected Jacobian is nearly singular, a small projected-residual perturbation can induce a large parameter deviation. The subsequent linear results make this distinction explicit.

\begin{corollary}[Stationary case]
\label{cor:stationary}
If $\boldsymbol{\theta}_t \equiv \boldsymbol{\theta}$ and 
$\sum_{t=1}^T \mathcal{L}^{\mathrm{KKT}}_t(\boldsymbol{\theta}') = 0$, then
\[
\boldsymbol{\theta}' = \boldsymbol{\theta}.
\]
\end{corollary}

\noindent\textit{Proof sketch.} Nonnegativity makes zero aggregate structural loss imply zero loss at every period, after which Theorem~\ref{thm:ident} applies pointwise. Complete proof: Online Supplement Section~S1.6.

\begin{proposition}[Consistency implies stationarity]
\label{prop:consistency-implies-stationarity}
If there exists $\bar{\boldsymbol{\theta}}$ such that 
$\mathcal{L}^{\mathrm{KKT}}_t(\bar{\boldsymbol{\theta}})=0$ for all $t$, then 
$\boldsymbol{\theta}_t$ is constant over time.
\end{proposition}

\noindent\textit{Proof sketch.} Applying Theorem~\ref{thm:ident} at each period forces the common zero-loss parameter to equal every $\boldsymbol{\theta}_t$. Complete proof: Online Supplement Section~S1.7.

\begin{theorem}[Linear identifiability]
\label{thm:linear-ident}
If 
\[
\nabla_{\mathbf{x}} c_t(\mathbf{x}_t;\boldsymbol{\theta})
=
\mathbf{A}_t(\mathbf{x}_t)\boldsymbol{\theta}
+
\mathbf{b}_t(\mathbf{x}_t),
\]
then identifiability holds if
\[
\widetilde{\mathbf{A}}_t := \Proj_t \mathbf{A}_t(\mathbf{x}_t)
\]
has full column rank.
\end{theorem}

\noindent\textit{Proof sketch.} After projection, the inverse condition is linear in $\boldsymbol{\theta}$; full column rank makes its null space trivial. Complete proof: Online Supplement Section~S1.8.

\begin{corollary}[Rank-based stability modulus]
\label{cor:rank}
If $\widetilde{\mathbf{A}}_t$ has smallest singular value 
$\sigma_{\min} > 0$, then
\[
\|\boldsymbol{\theta} - \boldsymbol{\theta}_t\|
\le
\frac{1}{\sigma_{\min}}
\left\|
\Proj_t \nabla_{\mathbf{x}} c_t(\mathbf{x}_t;\boldsymbol{\theta})
\right\|.
\]
\end{corollary}

\noindent\textit{Proof sketch.} The bound is the standard smallest-singular-value inequality for the projected linear system. Complete proof: Online Supplement Section~S1.9.

\begin{lemma}[Identifiability vs stability gap]
\label{lem:identifiability-gap}
Even when the inverse mapping is locally identifiable at each period,
independent per-period estimation can produce unstable parameter sequences under small perturbations.
\end{lemma}

\noindent\textit{Proof sketch.} Residual perturbations are amplified by $1/\sigma_{\min}(\widetilde A_t)$; without temporal coupling these pointwise errors can create excessive path variation. Complete proof: Online Supplement Section~S1.10.

The sequence of results separates three concepts that are easily conflated: local uniqueness, numerical conditioning, and temporal stability. Rank can guarantee the first two locally, but controlling a recovered path under repeated perturbations requires a temporal mechanism, motivating the offline coupling in Equation~\eqref{eq:trajectory-estimator} and the sequential update in Algorithm~\ref{alg:unified-strong}.

The practical implication is that a locally identifiable inverse model may still yield erratic estimates if its projected Jacobian is weakly conditioned. Period-by-period inversion then has no mechanism for sharing information across adjacent periods, so small local perturbations can appear as large artificial jumps in the recovered path. The trajectory formulation and the online recursion address this issue differently: the former penalizes path roughness explicitly, while the latter controls sequential loss through local first-order updates and comparator-based regret.

\subsection{Existence and Well-posedness}
\label{subsec:existence}

The next results establish existence and well-posedness of both the forward allocation problem and the 
inverse trajectory estimator, while identifying conditions that govern solution 
stability.

Regularity of the forward problem is ensured under the following assumptions.

\begin{assumption}[Level-boundedness / compactness]
\label{ass:level}
For each $t$, either the feasible set $\mathcal{X}_t$ is compact or the 
objective $c_t(\mathbf{x};\boldsymbol{\theta})$ is level-bounded in $\mathbf{x}$.
\end{assumption}

\begin{assumption}[Quadratic growth]
\label{ass:qg-main}
For each $t$, the forward problem satisfies
\[
c_t(\mathbf{x};\boldsymbol{\theta}_t) - c_t^\star(\boldsymbol{\theta}_t)
\ge
\frac{\alpha_t}{2}\|\mathbf{x} - \mathbf{x}_t^\star\|^2,
\]
for some $\alpha_t > 0$.
\end{assumption}

These conditions ensure existence and uniqueness of each forward solution under the stated assumptions.

Regularity of the inverse problem is captured as follows.

\begin{assumption}[Continuity and compactness of inverse problem]
\label{ass:theta-compact}
The parameter set $\Theta$ is compact, and the structural diagnostic $\mathcal{L}^{\mathrm{KKT}}_t(\boldsymbol{\theta})$ is lower semicontinuous.
\end{assumption}

\begin{theorem}[Existence and well-posedness]
\label{thm:exist-unique}
Under Assumptions~\ref{ass:level}--\ref{ass:theta-compact}:

\begin{enumerate}
\item For each $t$, the forward problem admits a unique minimizer 
$\mathbf{x}_t^\star$.

\item The inverse trajectory estimator
\[
\{\hat{\boldsymbol{\theta}}_t\}_{t=1}^T
\in
\arg\min_{\boldsymbol{\theta}_{1:T}\in\Theta^T}
\left\{
\sum_{t=1}^T \mathcal{L}^{\mathrm{KKT}}_t(\boldsymbol{\theta}_t)
+
\lambda \sum_{t=2}^T d(\boldsymbol{\theta}_t,\boldsymbol{\theta}_{t-1})
\right\}
\]
admits at least one solution.
\end{enumerate}
\end{theorem}

\noindent\textit{Proof sketch.} Continuity plus compactness or level-boundedness gives existence of the forward minimizer, quadratic growth gives uniqueness, and lower semicontinuity on compact $\Theta^T$ gives attainment of the offline inverse objective. Complete proof: Online Supplement Section~S1.11.

The theorem distinguishes the two optimization objects: the forward decision is unique under the stated growth condition, whereas the offline inverse trajectory is guaranteed to attain a minimum but need not be unique. Identifiability and conditioning from Section~\ref{subsec:identifiability} characterize local recoverability; temporal regularization discourages unnecessarily rough minimizing paths without creating uniqueness by itself.

This distinction matters for the subsequent analysis. Existence of an offline trajectory solution ensures that the regularized inverse benchmark is well defined, but it does not say that Algorithm~\ref{alg:unified-strong} converges to that batch minimizer. Conversely, the regret guarantees below concern the cumulative projected online loss and require no claim of uniqueness of the offline trajectory. Keeping these statements separate prevents a batch well-posedness result from being interpreted as a sequential convergence theorem.

\subsection{Static and Dynamic Regret}
\label{subsec:regret}

The regret analysis applies to Algorithm~\ref{alg:unified-strong} and its projected online loss $\ell_t$. The offline trajectory problem in Section~\ref{sec:model} is not used as the algorithmic update in these results.

\begin{assumption}[Online-loss regularity]
\label{ass:subgrad}
The parameter set $\Theta$ is convex with finite diameter $D$, each $\ell_t$ is convex on $\Theta$, and there exists $G>0$ such that every selected subgradient satisfies
\[
\|\mathbf{g}_t\|\le G,
\qquad
\mathbf{g}_t\in\partial \ell_t(\hat{\boldsymbol{\theta}}_t).
\]
The mirror map in Algorithm~\ref{alg:unified-strong} is strongly convex on $\Theta$. For the Euclidean specialization used in the proofs, the update reduces to projected online gradient descent with $\eta_t=\eta_0/\sqrt{t}$.
\end{assumption}

\begin{assumption}[Comparator-path variation]
\label{ass:drift}
For a comparator sequence $\{\boldsymbol{\theta}_t^\star\}_{t=1}^T\subset\Theta$, define
\[
V_T:=\sum_{t=2}^T\|\boldsymbol{\theta}_t^\star-\boldsymbol{\theta}_{t-1}^\star\|<\infty.
\]
When the true parameter path minimizes the period losses, it may be used as this comparator sequence.
\end{assumption}

\begin{theorem}[Static regret]
\label{thm:static-regret}
Under Assumption~\ref{ass:subgrad}, Algorithm~\ref{alg:unified-strong} with $\eta_t=\eta_0/\sqrt{t}$ satisfies
\[
\sum_{t=1}^T \ell_t(\hat{\boldsymbol{\theta}}_t)
-
\min_{\boldsymbol{\theta}\in\Theta}\sum_{t=1}^T\ell_t(\boldsymbol{\theta})
=O(\sqrt{T}).
\]
\end{theorem}

\noindent\textit{Proof sketch.} The standard projected online-gradient inequality telescopes; $\eta_t=\eta_0/\sqrt t$ yields the $O(\sqrt T)$ bound. Complete proof: Online Supplement Section~S1.12.

The static comparator is the best fixed parameter in hindsight; the result is a sequential optimization guarantee, not a statistical consistency statement.

\begin{theorem}[Dynamic regret]
\label{thm:dynamic-regret}
Under Assumptions~\ref{ass:subgrad}--\ref{ass:drift}, Algorithm~\ref{alg:unified-strong} with $\eta_t=\eta_0/\sqrt{t}$ satisfies
\[
\sum_{t=1}^T\ell_t(\hat{\boldsymbol{\theta}}_t)
-
\sum_{t=1}^T\ell_t(\boldsymbol{\theta}_t^\star)
=O\!\left(\sqrt{T}(1+V_T)\right).
\]
\end{theorem}

\noindent\textit{Proof sketch.} The same inequality is evaluated against a moving comparator, with successive comparator changes controlled by path length $V_T$. Complete proof: Online Supplement Section~S1.13.

The dynamic bound measures loss relative to a moving benchmark and does not imply pointwise recovery or equivalence with the offline trajectory optimizer. The schedule is horizon free.

%--------------------------------------------------
\subsection{Perturbation Stability}
\label{subsec:noise}

Observed decisions may be perturbed by measurement error, model mismatch, or approximate optimality. Rather than imposing a distributional rate that is not needed by the proof, the analysis uses a local deterministic residual bound.

\begin{assumption}[Local residual perturbation]
\label{ass:noise}
In the linear-gradient setting of Theorem~\ref{thm:linear-ident}, let
\[
\mathbf r_t(\boldsymbol{\theta})
:=\Proj_t\nabla_{\mathbf x}c_t(\mathbf x_t;\boldsymbol{\theta}),
\qquad
\widetilde{\mathbf A}_t:=\Proj_t\mathbf A_t(\mathbf x_t).
\]
Assume $\widetilde{\mathbf A}_t$ has full column rank and the true parameter satisfies
\[
\|\mathbf r_t(\boldsymbol{\theta}_t)\|\le \delta_t
\]
for a perturbation level $\delta_t\ge0$.
\end{assumption}

\begin{theorem}[Residual-to-parameter stability]
\label{thm:noise}
Under Assumption~\ref{ass:noise}, any candidate $\boldsymbol{\theta}$ satisfies
\[
\|\boldsymbol{\theta}-\boldsymbol{\theta}_t\|
\le
\frac{1}{\sigma_{\min}(\widetilde{\mathbf A}_t)}
\left(\|\mathbf r_t(\boldsymbol{\theta})\|+\delta_t\right).
\]
For the projected online loss $\ell_t(\boldsymbol{\theta})=\|\mathbf r_t(\boldsymbol{\theta})\|^2$, this becomes
\[
\|\boldsymbol{\theta}-\boldsymbol{\theta}_t\|
\le
\frac{\sqrt{\ell_t(\boldsymbol{\theta})}+\delta_t}
{\sigma_{\min}(\widetilde{\mathbf A}_t)}.
\]
\end{theorem}

\noindent\textit{Proof sketch.} The affine residual difference is bounded below by $\sigma_{\min}(\widetilde A_t)\|\boldsymbol{\theta}-\boldsymbol{\theta}_t\|$ and above by the triangle inequality. Complete proof: Online Supplement Section~S1.14.

The theorem separates residual fit from inverse conditioning. Small projected residual is informative only when $\sigma_{\min}(\widetilde{\mathbf A}_t)$ is not too small; poor conditioning can amplify even modest perturbations. The result is local and deterministic and should not be read as a statistical consistency or variance-rate theorem.

\section{Computational Experiments}
\label{sec:experiments}

\subsection{Experimental Design and Data-Generation Protocol}
\label{subsec:exp-design}

The computational study uses synthetic problems with known latent parameter paths so that structural recovery and tracking can be evaluated directly. The two stylized benchmarks are a multi-resource allocation problem motivated by capacity management and a network-dispatch problem motivated by energy systems. They differ in constraint structure, horizon, and cost composition but use the same inverse-learning mechanism. The terminology provides operational motivation rather than field calibration. For continuity with the stored computational outputs, compact table and plot labels use \emph{Healthcare} and \emph{Energy} as shorthand for the Multi-resource and Network-dispatch benchmarks, respectively.

The baseline latent path evolves through controlled increments $\boldsymbol{\theta}_t=\boldsymbol{\theta}_{t-1}+\delta_t$, and observed decisions satisfy $\tilde x_t=x_t^\star+\varepsilon_t$ with $\varepsilon_t\sim\mathcal N(0,\sigma^2I)$. Static, stronger-drift, and shock experiments modify this baseline as specified below. The online estimator uses $\ell_t(\theta)=\|\Proj_t\nabla_xc_t(\tilde x_t;\theta)\|^2$ and $\eta_t=\eta_0/\sqrt t$. Results are summarized across repeated seeds; where applicable, uncertainty is reported using a 95\% confidence interval (CI).

The experiments address six questions: noiseless recoverability, finite-horizon regret, sensitivity to gradual drift, observational perturbations, abrupt shocks, and the role of temporal coupling. They are finite-sample diagnostics, not proofs of asymptotic rates or of necessity of the sufficient rank condition. Full benchmark dimensions, forward objectives, horizons, noise grids, comparators, and implementation settings are reported in Online Supplement Tables~S1--S4.

The forward decision in each benchmark is the unique solution of
\[
x_t^\star=\arg\min_{x\in\mathcal X_t}c_t(x;\boldsymbol\theta_t),
\qquad
\mathcal X_t=\{x\ge0:B_tx\le q_t\}.
\]
The multi-resource objective combines a strongly convex allocation cost, latent priorities, and deviation from a reference allocation; the network-dispatch objective combines the strongly convex term with latent dispatch priorities and a linear dispatch penalty. Table~\ref{tab:benchmark-summary} gives the compact structural summary.

\begin{table}[!htbp]
\centering
\caption{Compact summary of the two stylized computational benchmarks. Detailed configurations and diagnostic targets are reported in Online Supplement Tables~S1--S4.}
\label{tab:benchmark-summary}
\scriptsize
\renewcommand{\arraystretch}{1.15}
\resizebox{\textwidth}{!}{%
\begin{tabular}{llll}
\toprule
Benchmark & $(n,k,p)$ & Parameter set & Forward objective \\
\midrule
Multi-resource & $(5,2,3)$ & $[0,5]^3$ & $\frac{\rho_h}{2}\|x\|^2-\beta_h(A\theta)^\top x+\lambda_f\|x-\bar{x}\|^2$ \\
Network-dispatch & $(4,1,3)$ & $[0,2]^3$ & $\frac{\rho_e}{2}\|x\|^2-\beta_e(A\theta)^\top x+\gamma e^\top x$ \\
\bottomrule
\end{tabular}%
}
\end{table}

The multi-resource benchmark uses two coupled capacity constraints and a shorter horizon, whereas network dispatch uses one aggregate resource constraint and a longer horizon. These differences create distinct temporal and conditioning regimes while retaining a common inverse structure. The complete data-generation and diagnostic specifications remain in the Online Supplement so that the main text focuses on the evidence needed to interpret the method.

For transparency, the two forward objectives are
\[
c_t^{(M)}(x;\theta)=\frac{\rho_h}{2}\|x\|^2-\beta_h(A\theta)^\top x+\lambda_f\|x-\bar x\|^2,
\]
and
\[
c_t^{(N)}(x;\theta)=\frac{\rho_e}{2}\|x\|^2-\beta_e(A\theta)^\top x+\gamma e^\top x.
\]
Both are strongly convex in the forward decision. The first provides two coupled resource constraints and a reference-allocation term; the second uses a single aggregate network constraint and a linear dispatch component. These choices are not intended as calibrated application models. Their purpose is to create distinct but controlled inverse geometries for testing the same estimator.

\subsubsection{Implementation and stored diagnostics}
The Euclidean specialization of Algorithm~\ref{alg:unified-strong} is used in the computational study with the horizon-free step size $\eta_t=\eta_0/\sqrt t$. Each run stores the ground-truth parameter path, optimal forward decisions, perturbed observations, estimated path, variation budget, inverse losses, tracking errors, and decision-level regret measures. Repeated random seeds are used for the reported averages and confidence intervals. Exact horizon, noise, initialization, and comparator settings are listed in Online Supplement Tables~S1--S4 so that these implementation details do not need to be repeated in every experimental subsection.

\subsection{Identifiability Diagnostic (Noiseless Setting)}
\label{subsec:identifiability-exp}

The noiseless diagnostic contrasts an identifiable construction satisfying the sufficient rank condition in Theorem~\ref{thm:linear-ident} with a deliberately rank-deficient construction. In the identifiable setting, recovered and true parameter trajectories coincide to machine precision; under rank deficiency, substantial errors remain despite the absence of observational noise.

Fig.~\ref{fig:identifiability} visualizes this contrast across the two stylized benchmarks, showing near-exact trajectory recovery in the identifiable construction and persistent deviations under rank deficiency.

\begin{figure}[t]
\centering
\includegraphics[width=\linewidth]{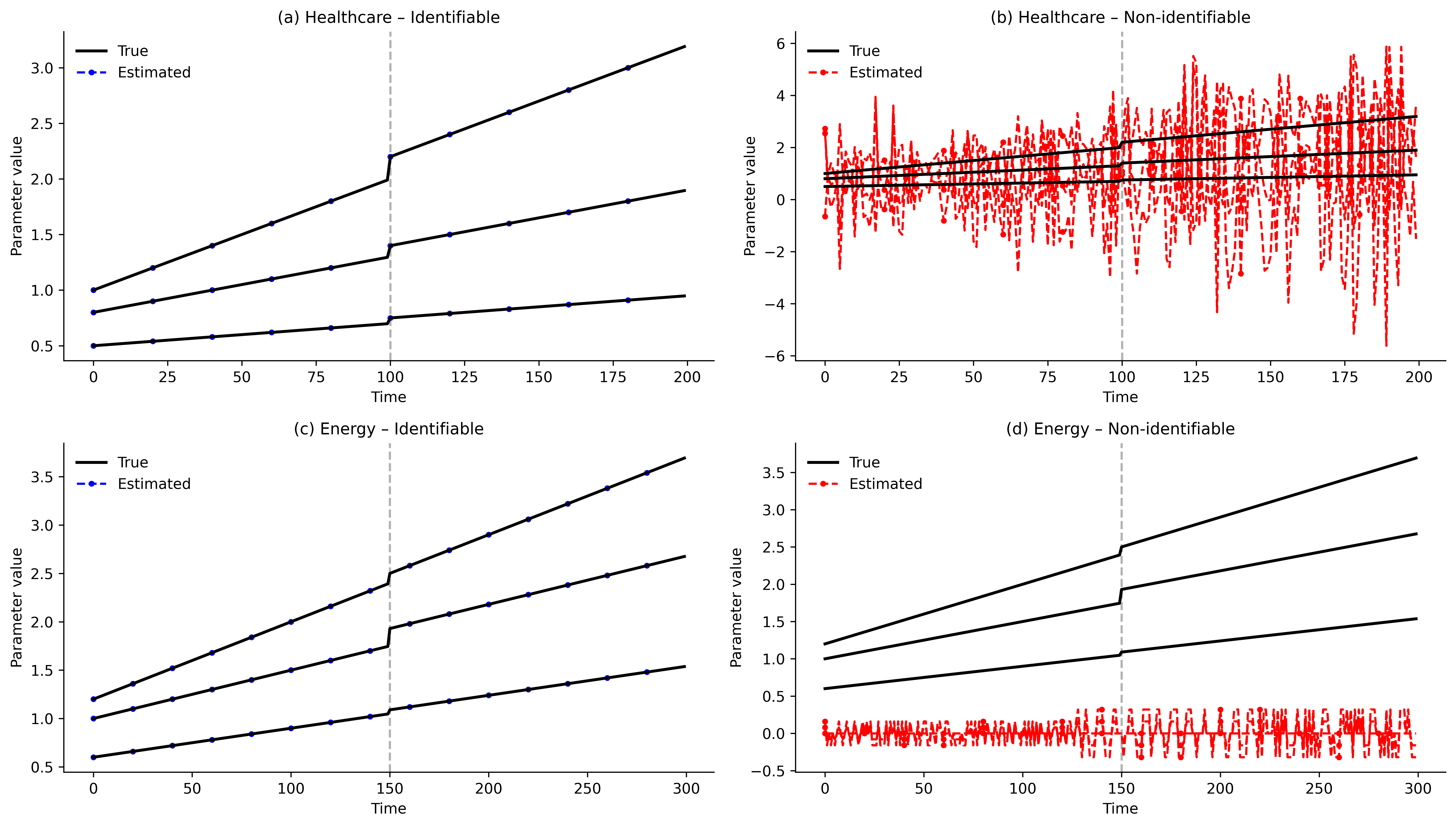}
\caption{
Identifiability diagnostic in the noiseless setting.
Solid black lines denote the true parameter trajectories, while dashed
colored lines with markers represent the estimated trajectories.
The left-side subfigures (identifiable case) show exact recovery, where the two trajectories overlap up to machine precision.
The right-side subfigures (rank-deficient construction) exhibit substantial deviations between the estimated and true trajectories.
}
\label{fig:identifiability}
\end{figure}

Across the two benchmarks, identifiable-case errors are on the order of $10^{-15}$, whereas the rank-deficient construction exhibits order-one recovery error. Table~\ref{tab:identifiability-summary} reports the exact ranges.

\begin{table}[t]
\centering
\caption{Error statistics for the identifiability diagnostic.}
\label{tab:identifiability-summary}
\begin{tabular}{lcc}
\toprule
Benchmark & Identifiable (min--max) & Rank-deficient (mean [min--max]) \\
\midrule
Multi-resource & $[1.7,\,4.3]\times10^{-15}$ & $2.37\,[0.91,\,9.09]$ \\
Network-dispatch & $[1.5,\,5.1]\times10^{-15}$ & $3.24\,[1.02,\,5.01]$ \\
\bottomrule
\end{tabular}
\end{table}

The machine-precision recovery accords with the sufficient rank condition in Theorem~\ref{thm:linear-ident}; the deliberately rank-deficient construction shows loss of recoverability despite zero observational noise. Because the theorem gives a sufficient condition, the comparison is not interpreted as establishing necessity. Subsequent experiments therefore shift from exact recovery to tracking error and regret under drift and perturbations.

\subsection{Static and Dynamic Regret Diagnostics}
\label{subsec:static-regret}

The regret experiment evaluates Algorithm~\ref{alg:unified-strong} under a static regime, bounded comparator variation, and increasing variation. Table~\ref{tab:regret_summary} and Fig.~\ref{fig:theory_validation} report the finite-horizon behavior relative to Theorems~\ref{thm:static-regret}--\ref{thm:dynamic-regret}.

\begin{table}[!htpb]
\centering
\caption{Summary of regret scaling and normalization behavior across domains.}
\label{tab:regret_summary}
\scriptsize
\begin{tabular}{l l c c c c c}
\toprule
Domain & Regime & $T$ & $V_T$ & Static Regret & Dynamic Regret & Normalized \\
\midrule
Healthcare & Static & 200 & 0.000 & 15.02 & 15.02 & 1.06 \\
Healthcare & Static & 800 & 0.000 & 39.81 & 39.81 & 1.41 \\
Healthcare & Dynamic ($V_T=\mathcal{O}(1)$) & 400 & 0.010 & 24.68 & 24.69 & 1.22 \\
Healthcare & Dynamic ($V_T=\mathcal{O}(\sqrt{T})$) & 800 & 0.148 & -- & 148.27 & 1.30 \\
\midrule
Energy & Static & 400 & 0.000 & 24.83 & 24.83 & 1.24 \\
Energy & Static & 1600 & 0.000 & 39.81 & 39.81 & 1.41 \\
Energy & Dynamic ($V_T=\mathcal{O}(1)$) & 800 & 0.010 & 24.68 & 24.68 & 1.22 \\
Energy & Dynamic ($V_T=\mathcal{O}(\sqrt{T})$) & 1600 & 0.757 & -- & 278.48 & 2.75 \\
\bottomrule
\end{tabular}
\end{table}

\begin{figure}[!htpb]
\centering
\includegraphics[width=\linewidth]{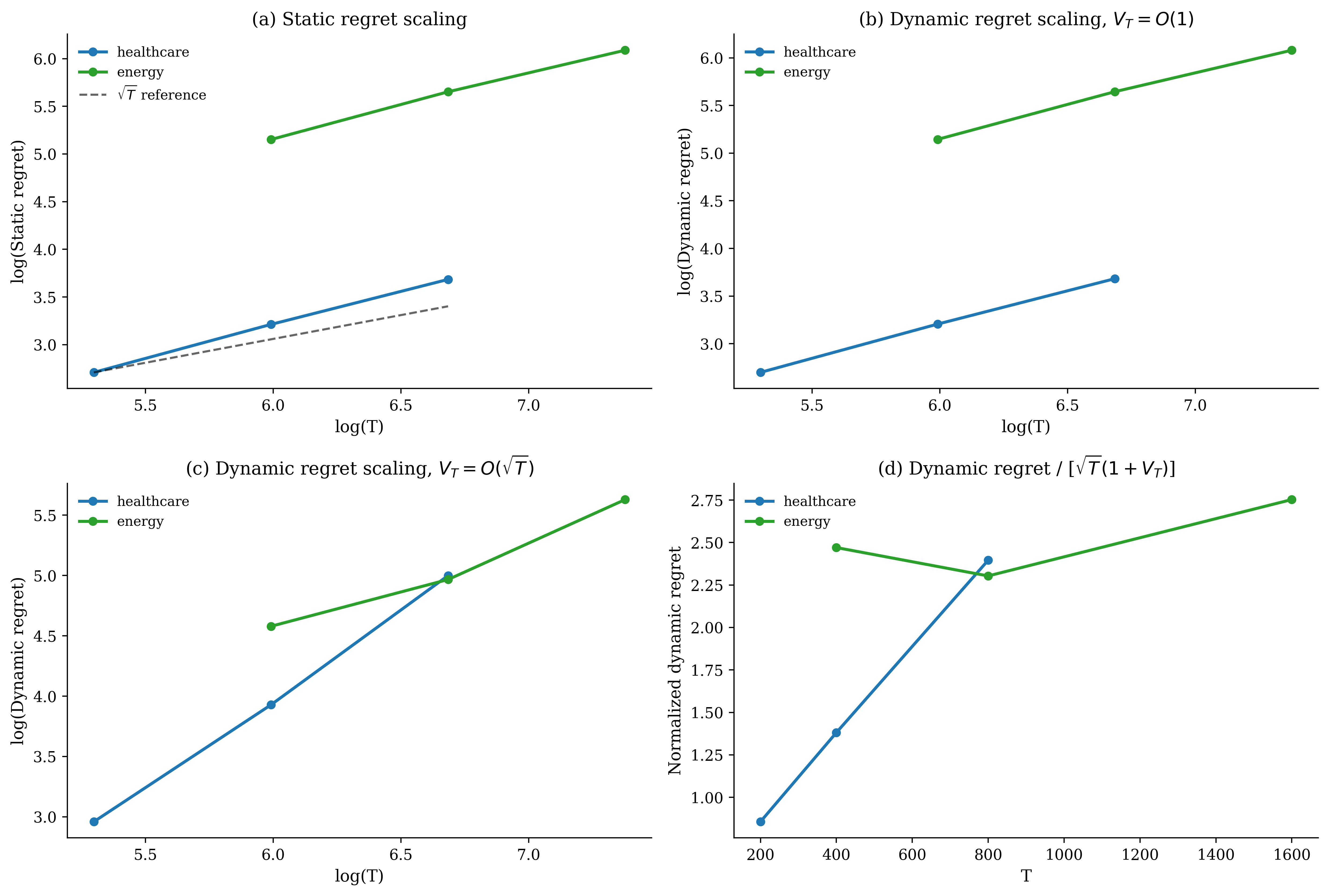}
\caption{
Finite-horizon regret diagnostics.
(a) Static regret across tested horizons with a $\sqrt{T}$ reference.
(b) Dynamic regret under bounded variation.
(c) Dynamic regret under increasing variation.
(d) Dynamic regret normalized by $\sqrt{T}(1+V_T)$.
}
\label{fig:theory_validation}
\end{figure}

Across the tested horizons, static and dynamic regret remain sublinear, while the normalized quantity $R_T/[\sqrt T(1+V_T)]$ stays bounded over the plotted grid. The static-regret log--log slopes are $0.703$ and $0.676$ in the two benchmarks; bounded-variation dynamic slopes are $0.709$ and $0.675$, and the stronger-variation network-dispatch slope is $0.758$. These slopes are finite-sample diagnostics, not estimates of the asymptotic exponent or evidence of tightness. The normalized values remain between $1.06$ and $1.41$ in the multi-resource benchmark and between $1.24$ and $2.75$ in network dispatch, showing no explosive growth over the tested horizons. Additional noise-stratified regret results and a naive-update comparison appear in Online Supplement Table~S5 and Figure~S1; the direct comparison of static, period-by-period, and trajectory-coupled estimators is reserved for Section~\ref{subsec:trajectory-necessity}.

\subsection{Drift Sensitivity and Variation Budget Scaling}
\label{subsec:drift-sensitivity}

To isolate temporal variation, the baseline path increments are multiplied by drift factors $\delta\in\{0.5,1.0,2.0\}$ in a noise-free setting. This produces increasing variation budgets $V_T$ while preserving the benchmark structure. Table~\ref{tab:drift-sensitivity} and Fig.~\ref{fig:drift-sensitivity} summarize the resulting tracking and regret diagnostics.

\begin{table}[t]
\centering
\caption{Drift sensitivity and variation-budget diagnostics.}
\label{tab:drift-sensitivity}
\scriptsize
\begin{tabular}{lcccccc}
\toprule
Benchmark & Drift & $V_T$ & Track. error & CI (95\%) & Inv. regret & Dec. regret \\
\midrule
Multi-resource & 0.5 & 0.915 & 0.900 & 0.207 & 163.30 & 443.61 \\
Multi-resource & 1.0 & 1.831 & 1.184 & 0.192 & 251.96 & 826.50 \\
Multi-resource & 2.0 & 3.661 & 1.788 & 0.293 & 680.14 & 1941.22 \\
\midrule
Network-dispatch & 0.5 & 0.915 & 1.784 & 0.144 & 1.09 & 625.72 \\
Network-dispatch & 1.0 & 1.831 & 1.961 & 0.181 & 36.79 & 1279.10 \\
Network-dispatch & 2.0 & 3.661 & 2.583 & 0.328 & 246.07 & 2410.42 \\
\bottomrule
\end{tabular}
\end{table}

\begin{figure}[t]
\centering
\includegraphics[width=\linewidth]{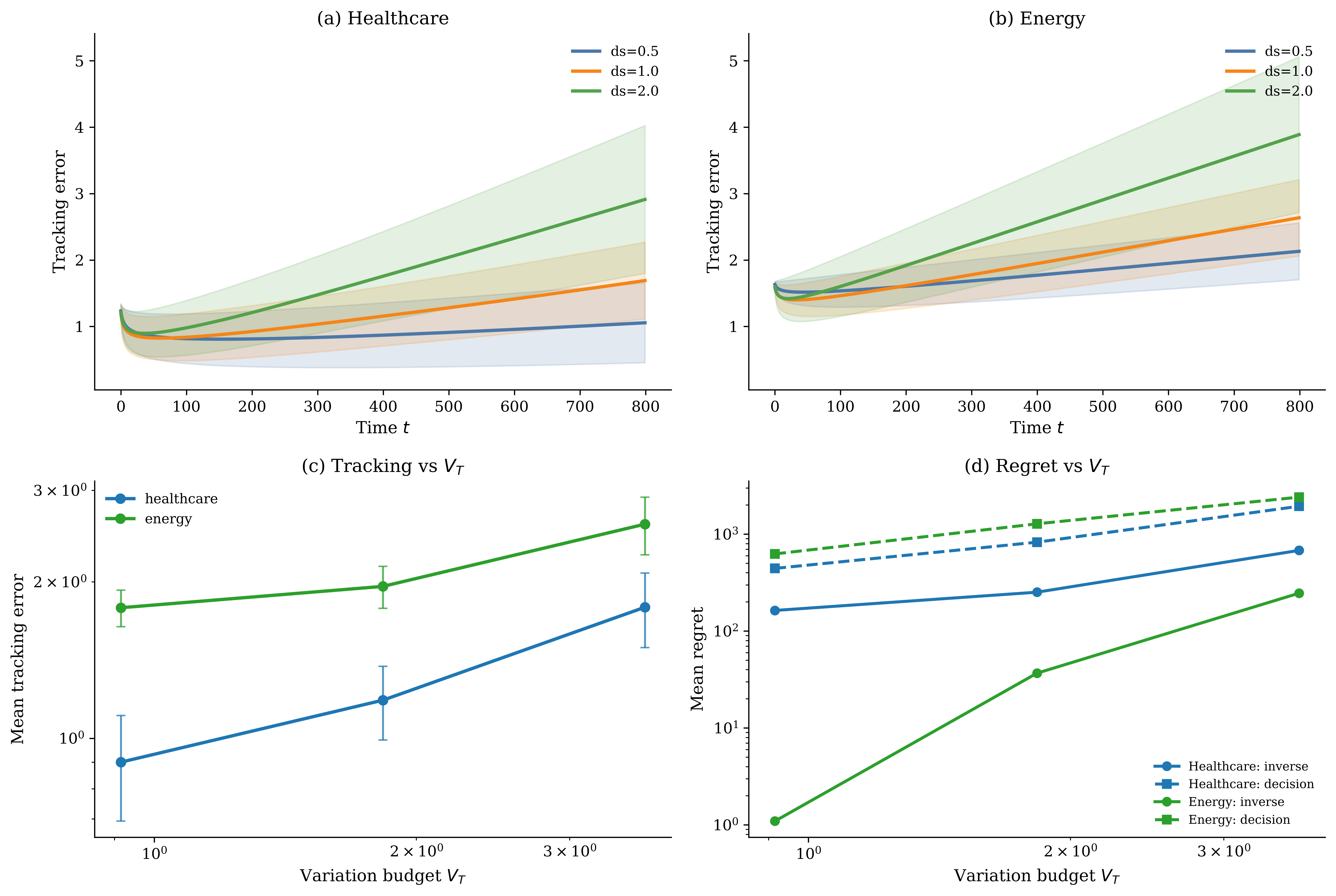}
\caption{
Drift sensitivity and variation budget scaling.
(a) Multi-resource: tracking trajectories under increasing drift.
(b) Network dispatch: tracking trajectories under increasing drift.
(c) Mean tracking error as a function of $V_T$.
(d) Scaling of inverse and decision regret with respect to $V_T$.
All results are obtained in a noise-free setting to isolate drift effects.
}
\label{fig:drift-sensitivity}
\end{figure}

Mean tracking error increases monotonically as $V_T$ rises from $0.915$ to $3.661$: from $0.90$ to $1.79$ in the multi-resource benchmark and from $1.78$ to $2.58$ in network dispatch. Inverse and decision regret increase over the same drift grid. The trajectories in Fig.~\ref{fig:drift-sensitivity}(a)--(b) show the same ordering over time, while Figure~\ref{fig:drift-sensitivity}(c)--(d) summarize tracking and regret against $V_T$. Because only three drift levels are tested, these subfigures are interpreted as monotone sensitivity diagnostics rather than evidence of a specific scaling law. The common within-benchmark pattern is consistent with comparator-path variation as a useful descriptor of nonstationary learning difficulty, while absolute error and regret levels remain benchmark specific.

The distinction between inverse regret and decision regret is also useful here. The former measures how well the evolving parameter explains the observed optimality signal, whereas the latter propagates parameter error through the forward decision problem. Their absolute scales therefore need not match across benchmarks. What is common across both domains is the directional response to increasing variation: greater path movement is associated with higher tracking error and larger cumulative losses. This is the empirical feature relevant to Theorem~\ref{thm:dynamic-regret}; cross-domain magnitude comparisons are not used to claim that one application is intrinsically more difficult than the other.

\subsection{Noise Sensitivity and Projection Effects}
\label{subsec:noise-robustness}

The noise study separates three effects: stationary observational noise, projection-induced attenuation of the residual signal, and the interaction of noise with drift. Detailed numerical results are reported in Online Supplement Tables~S6--S8, while Fig.~\ref{fig:noise-robustness} summarizes the main patterns.

\begin{figure}[!htpb]
\centering
\includegraphics[width=\linewidth]{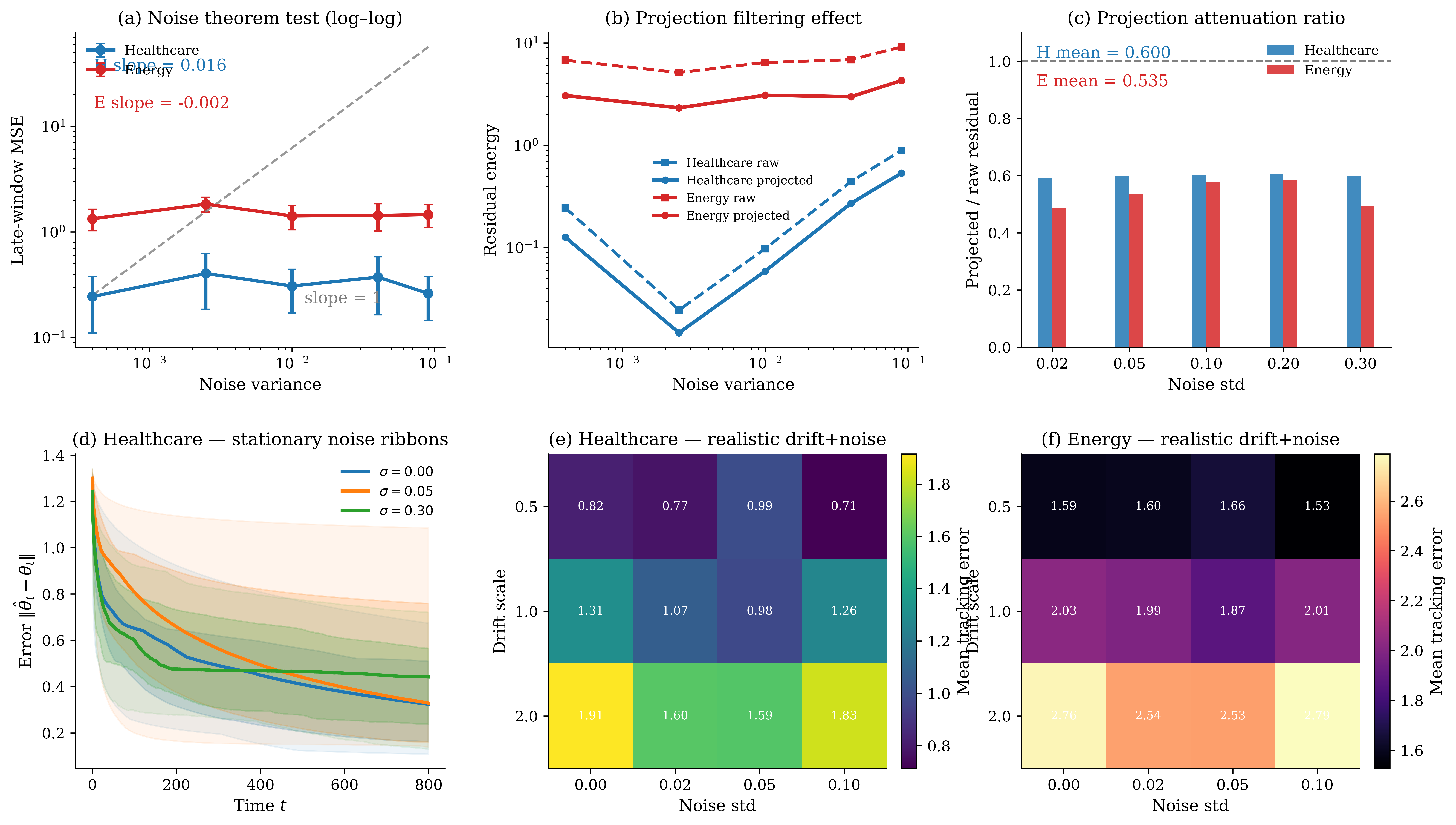}
\caption{
Noise-sensitivity and projection diagnostics.
(a) Stationary noise-sensitivity diagnostic, showing near-flat log--log behavior of late-window mean squared error (MSE) over the tested noise grid.
(b) Raw and projected residual energy, showing systematic reduction after projection.
(c) Attenuation ratio across tested noise levels.
(d) Tracking error trajectories under stationary observational noise.
(e) Multi-resource: tracking error under combined drift and noise.
(f) Network dispatch: tracking error under combined drift and noise.
}
\label{fig:noise-robustness}
\end{figure}

In the stationary grid, late-window MSE and tracking error change only modestly as nominal noise variance increases. Multi-resource late-window MSE remains between $0.257$ and $0.405$, with tracking error between $0.510$ and $0.602$; network-dispatch MSE remains between $1.418$ and $1.835$, with tracking error between $1.176$ and $1.369$. The near-flat empirical slopes in Fig.~\ref{fig:noise-robustness}(a) are therefore treated as diagnostics rather than as a claimed variance law.

Projection reduces residual energy in every reported setting. The projected-to-raw ratios are roughly $0.59$--$0.60$ in the multi-resource benchmark and $0.49$--$0.58$ in network dispatch. Figure~\ref{fig:noise-robustness}(b)--(c) show this attenuation directly, while Figure~\ref{fig:noise-robustness}(d) shows that stationary tracking trajectories retain similar shapes across the tested noise levels. This association is consistent with geometric filtering by $\Proj_t$, but it does not establish projection as the sole cause of the observed MSE pattern.

When drift and noise are varied jointly, changes along the drift dimension produce larger differences in tracking error over the tested grid than changes along the noise dimension. In the multi-resource benchmark, mean tracking error rises from $0.818$ to $1.915$ as drift increases from $0.5$ to $2.0$; in network dispatch it rises from $1.590$ to $2.765$. Figure~\ref{fig:noise-robustness}(e)--(f) visualize the same separation. This pattern is compatible with the theoretical distinction between comparator-path variation in Theorem~\ref{thm:dynamic-regret} and local residual sensitivity in Theorem~\ref{thm:noise}; it is not interpreted as a universal ordering of drift and noise effects.

The three diagnostics should be read jointly. A weak relationship between nominal noise variance and late-window MSE does not imply that observations are noise insensitive; Theorem~\ref{thm:noise} depends on the residual perturbation after the inverse geometry is applied, not directly on nominal $\sigma^2$. Likewise, attenuation of projected residual energy is a geometric observation rather than a causal decomposition of total estimation error. The combined grid is therefore used only to show how these mechanisms interact over the tested parameter range.

Fig.~\ref{fig:noise-robustness}(d) adds a temporal view that is not visible in the late-window summaries. All stationary-noise trajectories begin with a larger transient error and then settle into comparatively narrow bands. The similarity of these late trajectories is consistent with the modest differences reported in Online Supplement Table~S6, but it should not be read as a general robustness guarantee outside the tested grid. In particular, Theorem~\ref{thm:noise} makes the conditioning factor $1/\sigma_{\min}(\widetilde A_t)$ explicit, so a different constraint geometry can amplify a residual perturbation even when nominal observation noise is unchanged.

The projection diagnostic in Figure~\ref{fig:noise-robustness}(b)--(c) provides the corresponding geometric context. Since $\Proj_t$ removes components in $\mathrm{range}(B_t^\top)$, only the component compatible with feasible directions enters the projected online loss. The observed attenuation ratios below one are therefore structurally plausible, but the amount of attenuation is benchmark dependent. This is why the empirical discussion emphasizes association and conditioning rather than attributing all noise behavior to projection alone.

In the combined drift-and-noise grid, the heatmaps further show that the dominant direction of change over the tested range is aligned with the drift axis. This finite-grid observation complements, rather than replaces, the separate theoretical roles of $V_T$ and $\delta_t$: comparator movement enters the dynamic regret benchmark, while residual perturbation enters the local stability inequality.

\subsection{Shock Response under Structural Drift}
\label{subsec:shock}

Abrupt changes in the latent path create localized increases in $V_T$ and therefore test adaptation beyond the smooth-drift regime. Unlike observational noise, these shocks alter the latent parameter path itself; accordingly, they enter the experiment through comparator-path variation rather than solely through perturbations of the observed inverse signal. Table~\ref{tab:shock-key} reports the key response measures; the complete decomposition of baseline error, residual ratios, confidence intervals, and cross-benchmark tests is reported in Online Supplement Table~S9.

\begin{table}[!htbp]
\centering
\caption{Key shock-response diagnostics.}
\label{tab:shock-key}
\scriptsize
\begin{tabular}{lcccc}
\toprule
Benchmark & Shock & Peak uplift & Recovery time & Post-shock burden \\
\midrule
Multi-resource & 0.5 & 0.049 & 0 & 20.94 \\
Multi-resource & 1.0 & 0.088 & 0 & 27.96 \\
Multi-resource & 2.0 & 0.178 & 145.9 & 39.56 \\
Network-dispatch & 0.5 & 0.067 & 0 & 36.73 \\
Network-dispatch & 1.0 & 0.119 & 0 & 52.35 \\
Network-dispatch & 2.0 & 0.204 & $>200$ & 64.17 \\
\bottomrule
\end{tabular}
\end{table}

\begin{figure*}[!htbp]
\centering
\includegraphics[width=\textwidth]{Figure5.png}
\caption{
Shock response under structural drift.
(a) Healthcare trajectories with quantile ribbons.
(b) Energy trajectories with quantile ribbons.
(c) Recovery phase diagram. The horizontal axis represents shock scale and the vertical axis represents residual-to-peak ratio. Color indicates recovery time, and marker shape distinguishes domains (circle: Healthcare, square: Energy). The dashed horizontal line indicates the persistence threshold (ratio = 1).
(d) Peak--residual decomposition of excess error.
(e) Healthcare excess-error heatmap.
(f) Energy excess-error heatmap.
}
\label{fig:shock}
\end{figure*}

Peak uplift and cumulative post-shock burden increase monotonically with shock magnitude in both benchmarks. At shock scale $2.0$, peak uplift reaches $0.178$ in the multi-resource benchmark and $0.204$ in network dispatch, while burden reaches $39.56$ and $64.17$, respectively. The complete decomposition in Online Supplement Table~S9 further shows that the residual-to-peak ratio falls below one for the strongest multi-resource shock but remains above one for every network-dispatch shock. Thus, the strongest multi-resource shock shows delayed but finite recovery (recovery time $145.9$), whereas network dispatch does not meet the recovery criterion within the 200-step evaluation window.

Fig.~\ref{fig:shock} separates transient amplification from persistence. Figure~\ref{fig:shock}(a)--(b) show the immediate divergence at shock onset, Figure~\ref{fig:shock}(c) places residual persistence and recovery time in a common phase diagram, and Figure~\ref{fig:shock}(d)--(f) separate peak, residual, and time-localized excess error. Larger shocks produce larger immediate deviations, but recovery is not determined by peak size alone; the residual-to-peak ratio and cumulative burden distinguish slow recovery from persistent displacement. The experiment therefore extends the variation-budget interpretation to discontinuous nonstationarity while preserving the distinction between structural shocks and observational perturbations.

This distinction is especially visible at the largest shock. In the multi-resource benchmark, the residual ratio falls to $0.932$, so the eventual error is below the transient peak even though recovery is slow. In network dispatch, the corresponding residual ratio remains $1.277$, consistent with persistent displacement over the evaluation window. Hence recovery time, residual ratio, and cumulative burden measure different aspects of post-shock behavior and should not be collapsed into a single recovery statistic.

The phase diagram in Fig.~\ref{fig:shock}(c) is designed precisely to separate these dimensions. A long recovery time can coexist with a residual ratio below one, as in the strongest multi-resource shock, whereas a ratio above one indicates that the post-shock deviation remains at least as large as the initial peak component. The heatmaps in Figure~\ref{fig:shock}(e)--(f) add the time localization: the multi-resource benchmark shows a concentrated region of excess error that dissipates, while network dispatch retains a broader elevated region.

These patterns also clarify why the shock experiment is not interchangeable with the noise experiment. Observation noise perturbs the signal used to infer the parameter, whereas a structural shock changes the target parameter itself and therefore contributes directly to comparator-path variation. The two mechanisms can produce superficially similar increases in tracking error but enter the analysis through different quantities. Treating them separately keeps the empirical interpretation aligned with the theory.

\subsection{Why Trajectory Coupling Matters}
\label{subsec:trajectory-necessity}

Local identifiability does not answer whether temporal coupling improves a recovered path. We therefore compare the proposed trajectory-coupled estimator with static inverse optimization and independent per-period inversion under drift, noise, and a structural shock.

Fig.~\ref{fig:trajectory_necessity} and Table~\ref{tab:trajectory_necessity} summarize this comparison at both the trajectory and decision levels, including tracking error, cumulative decision regret, post-shock burden, and estimated-to-true path variation.

\begin{figure}[!htpb]
\centering
\includegraphics[width=\linewidth]{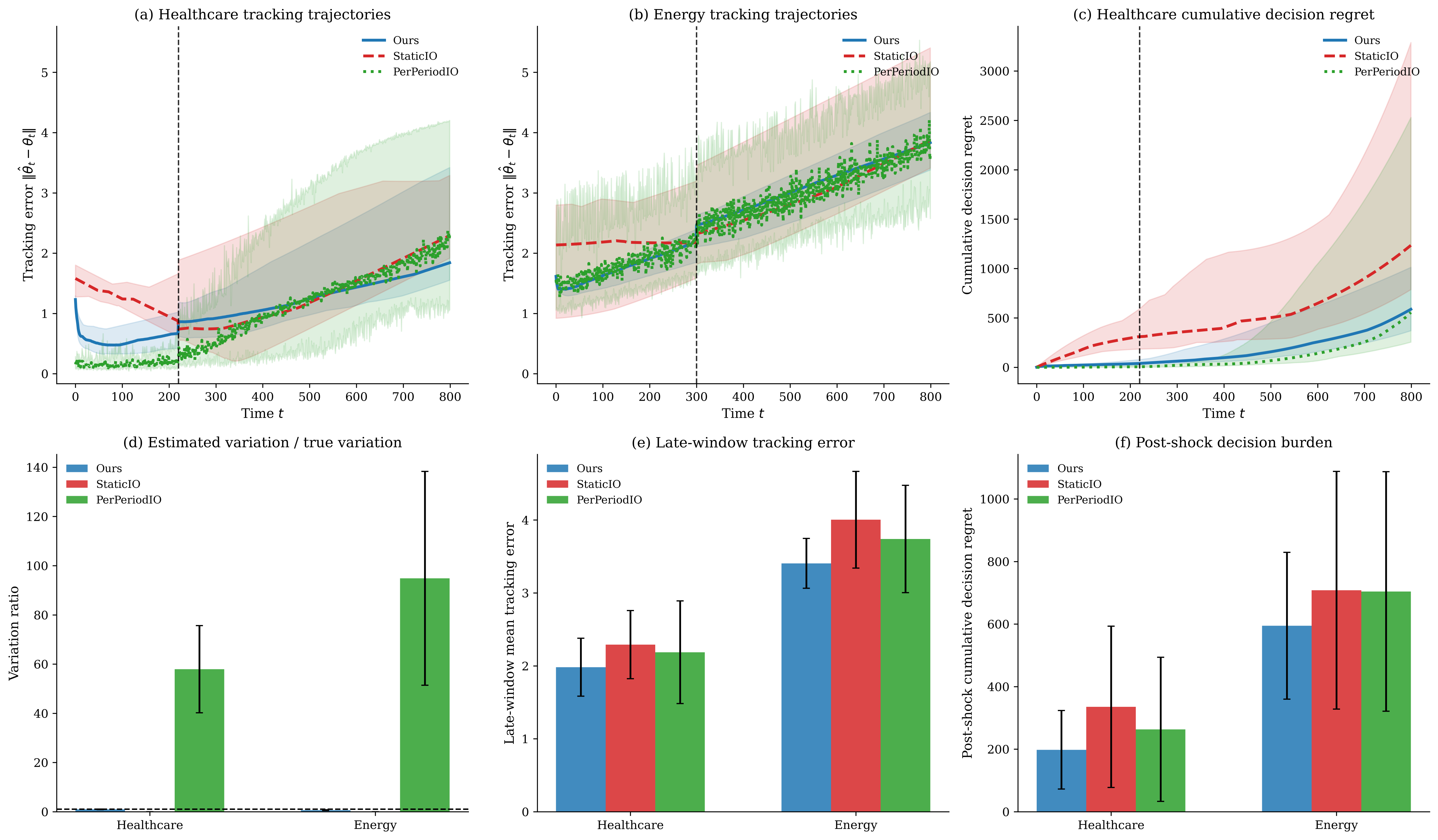}
\caption{
Comparison of trajectory-coupled and decoupled estimators under drift, noise, and shock.
(a)--(b) Tracking error over time.
(c) Cumulative decision regret.
(d) Estimated-to-true variation ratio.
(e) Late-stage tracking error.
(f) Post-shock decision burden.
Vertical dashed lines indicate the shock time.
}
\label{fig:trajectory_necessity}
\end{figure}

\begin{table}[!htbp]
\centering
\caption{Comparison of trajectory-coupled and decoupled estimators under drift, noise, and structural shocks}
\label{tab:trajectory_necessity}
\scriptsize
\setlength{\tabcolsep}{3pt}
\renewcommand{\arraystretch}{1.1}
\begin{tabular}{llcccc}
\toprule
Domain & Method & Tracking Error & Decision Regret & Post-shock Burden & Variation Ratio \\
\midrule

\multirow{3}{*}{Healthcare}
& \makecell[l]{Trajectory-coupled\\(Proposed)} & $1.32$ & $1215$ & $197$ & $0.96$ \\
& Per-period IO & $1.30$ & $1917$ & $263$ & $57.91$ \\
& Static IO & $1.78$ & $2480$ & $335$ & $0.00$ \\

\midrule

\multirow{3}{*}{Energy}
& \makecell[l]{Trajectory-coupled\\(Proposed)} & $2.51$ & $2503$ & $594$ & $0.64$ \\
& Per-period IO & $3.08$ & $2998$ & $704$ & $94.86$ \\
& Static IO & $3.17$ & $3447$ & $708$ & $0.00$ \\

\bottomrule
\end{tabular}

\vspace{4pt}
\begin{minipage}{0.95\linewidth}
\footnotesize
\textit{Notes:} Tracking error is averaged over all time periods. Decision regret is cumulative over the full horizon. 
Post-shock burden measures cumulative regret within a fixed window following the structural shock. 
Variation ratio is defined as the ratio between estimated temporal variation and true variation. 
A value close to one indicates estimated path variation comparable to the true path variation, whereas very large values indicate an overly variable estimated path.
\end{minipage}
\end{table}

Per-period inversion can match the proposed method on pointwise tracking yet produces dramatically greater path variation: the estimated-to-true variation ratio exceeds $57$ and $94$ in the two benchmarks, compared with $0.96$ and $0.64$ for the proposed method. Static inversion has zero estimated variation and therefore cannot adapt to sustained change. The decision-level metrics expose the corresponding trade-off. In the multi-resource benchmark, cumulative decision regret is $1215$ for the proposed method versus $1917$ for per-period IO and $2480$ for static IO; in network dispatch the corresponding values are $2503$, $2998$, and $3447$. Post-shock burden follows the same ordering.

Paired tests show that tracking differences are not uniformly significant, whereas variation differences are pronounced ($p<10^{-6}$). This is important for interpretation: trajectory coupling is not claimed to dominate every alternative on every instantaneous metric. Its principal empirical role is stabilization of the recovered path and the downstream decisions generated from it. The result directly complements Lemma~\ref{lem:identifiability-gap}: local uniqueness can coexist with an unstable period-by-period path, while coupling moderates the extremes of under-adaptation and excessive responsiveness.

The three estimators therefore represent distinct temporal assumptions. Static IO imposes complete persistence, per-period IO imposes none, and the proposed method allows adaptation while discouraging unnecessary movement. The variation ratio makes this distinction observable independently of tracking error. A value near one indicates that estimated path movement is of the same order as the true path movement; a very large value indicates that local perturbations are being translated into artificial parameter volatility. This comparison is the computational counterpart of the theoretical separation between identifiability and trajectory stability.

\section{Discussion}
\label{sec:implications}

The experiments clarify the distinction between three questions that can otherwise be conflated in dynamic inverse optimization: whether a latent parameter is locally identifiable, whether its recovery is well conditioned, and whether an estimator can track a changing parameter sequence without excessive temporal variation. The numerical results are interpreted through this distinction rather than as direct proofs of the theoretical bounds.

\paragraph{Identifiability and conditioning.}
The noiseless experiment in Section~\ref{subsec:identifiability-exp} shows machine-precision recovery for the construction satisfying the sufficient projected-rank condition and loss of recoverability for a deliberately rank-deficient construction. This comparison illustrates the role of local inverse geometry. It does not imply that the rank condition is necessary in every model, nor that local uniqueness guarantees robustness to perturbations. Corollary~\ref{cor:rank} and Lemma~\ref{lem:identifiability-gap} make the latter point explicit: conditioning, not injectivity alone, determines how residual perturbations translate into parameter error.

\paragraph{Temporal coupling and sequential estimation.}
Section~\ref{subsec:trajectory-necessity} offers the most direct comparison of the three temporal responses considered in the study. Static inverse optimization yields zero estimated path variation because it does not adapt. Independent period-by-period inversion can fit local observations aggressively, but in the reported experiments it produces variation far above the true path variation. The trajectory-aware estimator lies between these extremes. Its role is therefore better described as controlled adaptation than as universal superiority. This interpretation also clarifies the purpose of the online recursion: it supplies a streaming estimator with regret guarantees, not an identifiability device and not an exact solver for the offline regularized path problem.

\paragraph{Drift, noise, and shocks.}
Increasing path variation is associated with larger tracking error and cumulative regret in both stylized benchmarks. Observational-noise experiments show weaker changes over the tested noise grid, while the projected residual has lower energy than its unprojected counterpart. These observations suggest that projection can attenuate perturbations before they enter the inverse update, but the experiments do not identify projection as the sole cause of the observed error pattern. Shock experiments add a separate point: abrupt changes in the latent trajectory generate localized increases in tracking error, with recovery behavior differing across the two benchmark structures.

\paragraph{Scope of the operational interpretation.}
The healthcare and energy settings are stylized allocation benchmarks rather than calibrated field studies. Their purpose is to expose structural differences in horizon length, constraints, drift, and shock response while retaining known ground-truth parameters. Consequently, managerial conclusions should be limited to the mechanisms represented by the models. Claims about clinical practice, grid operation, or policy effectiveness would require domain-calibrated data and additional validation.

\paragraph{Limitations.}
Several limitations define the scope of the results. The regret analysis relies on convexity and boundedness conditions appropriate for the projected online loss. The identifiability results are local, and the strongest rank interpretation is tied to the linear-gradient setting. The perturbation theorem is a local residual-to-parameter stability result rather than a full statistical consistency theorem. Finally, the computational evidence is synthetic and finite-horizon. These restrictions are deliberate: they keep the theoretical statements aligned with what is proved and the experiments aligned with what is observed.

The analysis supports a view of dynamic inverse optimization in which structural recovery and sequential adaptation are complementary but distinct. Identifiability determines whether a local inverse explanation is unique; conditioning determines sensitivity; and online estimation determines how a changing target is followed over time.

\section{Conclusion}
\label{sec:conclusion}

This work studies recovery of time-varying preference parameters from observed allocation decisions by connecting inverse structure with sequential learning. The proposed formulation separates a structural inverse diagnostic from a projected online loss, establishes sufficient conditions for local identifiability and existence of trajectory solutions, and develops a sequential mirror-descent estimator with static regret $O(\sqrt{T})$ and dynamic regret $O(\sqrt{T}(1+V_T))$. A local residual-to-parameter bound further links recovery error to the conditioning of the projected inverse map. Controlled experiments on two stylized allocation benchmarks complement the theory: the noiseless construction illustrates recoverability under the sufficient rank condition, finite-horizon experiments exhibit sublinear regret behavior over the tested horizons, and drift, shock, and baseline-comparison experiments show how stronger temporal variation increases tracking difficulty and why temporal coupling can improve cumulative decision performance relative to static or period-by-period inversion. These findings are interpreted as numerical diagnostics rather than as proofs of asymptotic rates or universal dominance.

The analysis also clarifies the distinct roles of identifiability, conditioning, and sequential adaptation. Local uniqueness does not guarantee favorable conditioning, nor does it regulate how independently recovered parameters vary across periods; the online estimator addresses this separate temporal problem through recursive updating and regret-based adaptation to a moving target. Future work can extend the framework to partial observability, domain-calibrated operational data, nonlinear and mixed-integer forward models, uncertain feasible sets or behavioral models, and statistical uncertainty quantification for the recovered parameter path beyond the deterministic local stability bound considered here. Overall, the central contribution is to connect recoverability, sensitivity, and temporal adaptation within a unified dynamic inverse-optimization framework.

% =========================
\section*{Statements and Declarations}
\textbf{Funding:} This research received no external funding.\\
\textbf{Competing Interests:} The author declares no competing interests.\\
\textbf{Author Contributions:} The author solely conducted the conceptualization,
methodology, formal analysis, computational experiments, software implementation,
validation, and manuscript preparation.\\
\textbf{Data and Code Availability:}
The study relies exclusively on synthetic data. The synthetic datasets,
computational outputs, reproducibility materials, and simulation code
supporting this study are publicly available in Zenodo at
\url{https://doi.org/10.5281/zenodo.22289993}.\\
\textbf{Biographical Note:}
Jinho Cha is a faculty member in the Department of Computer Science at Gwinnett Technical College in Lawrenceville, Georgia, USA. He previously served as an Assistant Professor of Mathematics at the Korea Military Academy and as a researcher at the Korea Research Institute for Defense Technology Planning and Advancement, Republic of Korea. He earned his Ph.D. in Industrial Engineering from Clemson University and M.S. degrees in Industrial Engineering from Texas A\&M University and the University of Florida. His research interests include inverse optimization, online optimization, operations research, data-driven decision-making, and AI-enabled decision systems.\\

\bibliographystyle{apalike}
\spacingset{0.8}
\bibliography{references}

\clearpage

\end{document}

\sloppy
\def\spacingset#1{\renewcommand{\baselinestretch}{#1}\small\normalsize}
\spacingset{1.25}

\begin{center}
{\LARGE\bf Online Supplement}\\[0.75em]
{\large\bf Dynamic Inverse Optimization under Drift and Shocks: Theory, Regret Bounds, and Applications}
\end{center}
\bigskip

\noindent
This Online Supplement contains the complete proofs supporting the theoretical results and additional computational specifications and diagnostics referenced in the main manuscript. References to results in the main manuscript are stated explicitly by theorem, lemma, assumption, algorithm, and section number so that this supplement can be compiled independently.
\bigskip

\section{Proofs for Section 4: Theoretical Results}
\label{app:proofs}

\subsection{Preliminaries}
\label{app:proof:preliminaries}

Several standard linear-algebraic identities used throughout the proofs are collected here, with particular emphasis on projection operators and the Moore--Penrose pseudoinverse.

For any matrix $\mathbf{M}$, let $\mathbf{M}^\dagger$ denote its Moore--Penrose 
pseudoinverse. The following properties will be used repeatedly:
\begin{align}
(\mathbf{M}^\dagger)^\top &= (\mathbf{M}^\top)^\dagger, \\
\mathbf{M}\mathbf{M}^\dagger \mathbf{M} &= \mathbf{M}, \\
\mathbf{M}^\dagger \mathbf{M} \mathbf{M}^\dagger &= \mathbf{M}^\dagger.
\end{align}
Moreover, both $\mathbf{M}\mathbf{M}^\dagger$ and $\mathbf{M}^\dagger \mathbf{M}$ 
are symmetric and idempotent matrices, and represent orthogonal projections onto 
$\mathsf{range}(\mathbf{M})$ and $\mathsf{range}(\mathbf{M}^\top)$, respectively.

For each period $t$, define
\begin{equation}
\label{eq:proj-def}
\Proj_t := \mathbf{I} - \mathbf{B}_t^\top (\mathbf{B}_t \mathbf{B}_t^\top)^\dagger \mathbf{B}_t.
\end{equation}

This definition does not require $\mathbf{B}_t$ to have full row rank, and therefore 
covers both full-rank and rank-deficient cases in a unified manner.

The proofs repeatedly use the orthogonal decomposition
\begin{equation}
\label{eq:orthogonal-decomposition}
\mathbb{R}^n = \mathsf{ker}(\mathbf{B}_t) \oplus \mathsf{range}(\mathbf{B}_t^\top),
\end{equation}
where $\mathsf{ker}(\mathbf{B}_t)$ and $\mathsf{range}(\mathbf{B}_t^\top)$ are orthogonal 
complements. In particular,
\begin{equation}
\mathsf{ker}(\mathbf{B}_t) = \left(\mathsf{range}(\mathbf{B}_t^\top)\right)^\perp.
\end{equation}

The operator $\Proj_t$ will be shown to be the orthogonal projector onto 
$\mathsf{ker}(\mathbf{B}_t)$, and plays a central role in linking the structural Karush--Kuhn--Tucker (KKT) diagnostic
to the corresponding optimality conditions.

\subsection{Proof of Lemma~1 (Projector properties)}
\label{app:proof:proj-properties}

\begin{proof}
Recall that
\[
\Proj_t := \mathbf{I} - \mathbf{B}_t^\top (\mathbf{B}_t \mathbf{B}_t^\top)^\dagger \mathbf{B}_t.
\]

To verify that $\Proj_t$ is the orthogonal projector onto $\mathsf{ker}(\mathbf{B}_t)$, it is enough to establish symmetry, idempotence, and the characterization of its range.

\medskip
\noindent
\textbf{Step 1: Symmetry.}

Since $\mathbf{B}_t \mathbf{B}_t^\top$ is symmetric positive semidefinite, its Moore--Penrose pseudoinverse 
is also symmetric. Hence,
\[
\Proj_t^\top
=
\mathbf{I}
-
\mathbf{B}_t^\top \big( (\mathbf{B}_t \mathbf{B}_t^\top)^\dagger \big)^\top \mathbf{B}_t
=
\mathbf{I}
-
\mathbf{B}_t^\top (\mathbf{B}_t \mathbf{B}_t^\top)^\dagger \mathbf{B}_t
=
\Proj_t.
\]

Thus, $\Proj_t$ is symmetric.

\medskip
\noindent
\textbf{Step 2: Idempotence.}

Expanding $\Proj_t^2$ gives
\[
\Proj_t^2
=
\mathbf{I}
-2\mathbf{B}_t^\top (\mathbf{B}_t \mathbf{B}_t^\top)^\dagger \mathbf{B}_t
+
\mathbf{B}_t^\top (\mathbf{B}_t \mathbf{B}_t^\top)^\dagger
(\mathbf{B}_t \mathbf{B}_t^\top)
(\mathbf{B}_t \mathbf{B}_t^\top)^\dagger \mathbf{B}_t.
\]

Using the Moore--Penrose identity
\[
(\mathbf{B}_t \mathbf{B}_t^\top)^\dagger
(\mathbf{B}_t \mathbf{B}_t^\top)
(\mathbf{B}_t \mathbf{B}_t^\top)^\dagger
=
(\mathbf{B}_t \mathbf{B}_t^\top)^\dagger,
\]
the last term reduces to $\mathbf{B}_t^\top (\mathbf{B}_t \mathbf{B}_t^\top)^\dagger \mathbf{B}_t$, 
and therefore
\[
\Proj_t^2 = \Proj_t.
\]

\medskip
\noindent
\textbf{Step 3: $\mathsf{range}(\Proj_t) \subseteq \mathsf{ker}(\mathbf{B}_t)$.}

For any $\mathbf{v} \in \mathbb{R}^n$,
\[
\mathbf{B}_t \Proj_t \mathbf{v}
=
\mathbf{B}_t \mathbf{v}
-
\mathbf{B}_t \mathbf{B}_t^\top (\mathbf{B}_t \mathbf{B}_t^\top)^\dagger \mathbf{B}_t \mathbf{v}.
\]

Since $\mathbf{B}_t \mathbf{v} \in \mathsf{range}(\mathbf{B}_t)$ and 
\[
\mathbf{B}_t \mathbf{B}_t^\top (\mathbf{B}_t \mathbf{B}_t^\top)^\dagger
\]
is the orthogonal projector onto $\mathsf{range}(\mathbf{B}_t)$, it follows that
\[
\mathbf{B}_t \Proj_t \mathbf{v} = 0.
\]

Hence $\Proj_t \mathbf{v} \in \mathsf{ker}(\mathbf{B}_t)$.

\medskip
\noindent
\textbf{Step 3.5: Annihilation of $\mathbf{B}_t^\top$.}

Direct calculation gives
\[
\Proj_t \mathbf{B}_t^\top
=
\mathbf{B}_t^\top
-
\mathbf{B}_t^\top (\mathbf{B}_t \mathbf{B}_t^\top)^\dagger \mathbf{B}_t \mathbf{B}_t^\top.
\]

Since $(\mathbf{B}_t \mathbf{B}_t^\top)^\dagger(\mathbf{B}_t \mathbf{B}_t^\top)$ is the orthogonal projector onto $\mathsf{range}(\mathbf{B}_t)$, the Moore--Penrose identities give
\[
\mathbf{B}_t^\top(\mathbf{B}_t \mathbf{B}_t^\top)^\dagger
(\mathbf{B}_t \mathbf{B}_t^\top)=\mathbf{B}_t^\top.
\]
Therefore,
\[
\Proj_t \mathbf{B}_t^\top = 0.
\]

\medskip
\noindent
\textbf{Step 4: $\mathsf{ker}(\mathbf{B}_t) \subseteq \mathsf{range}(\Proj_t)$.}

Let $\mathbf{u} \in \mathsf{ker}(\mathbf{B}_t)$. Then $\mathbf{B}_t \mathbf{u} = 0$, and thus
\[
\Proj_t \mathbf{u}
=
\mathbf{u}
-
\mathbf{B}_t^\top (\mathbf{B}_t \mathbf{B}_t^\top)^\dagger \mathbf{B}_t \mathbf{u}
=
\mathbf{u}.
\]

Therefore, every vector in $\mathsf{ker}(\mathbf{B}_t)$ is fixed by $\Proj_t$.

\medskip
\noindent
\textbf{Step 5: Conclusion.}

Combining Steps 3 and 4, it follows that
\[
\mathsf{range}(\Proj_t) = \mathsf{ker}(\mathbf{B}_t).
\]

Moreover, since
\[
\mathbb{R}^n = \mathsf{ker}(\mathbf{B}_t) \oplus \mathsf{range}(\mathbf{B}_t^\top),
\]
and $\Proj_t$ annihilates $\mathsf{range}(\mathbf{B}_t^\top)$ while acting as the identity 
on $\mathsf{ker}(\mathbf{B}_t)$, it follows that $\Proj_t$ is the orthogonal projector 
onto $\mathsf{ker}(\mathbf{B}_t)$.
\end{proof}
This decomposition plays a central role in linking the structural KKT diagnostic to KKT optimality 
conditions, as it isolates gradient components orthogonal to feasible directions.

\subsection{Proof of Lemma~2 (KKT equivalence)}
\label{app:proof:kkt-equivalence}

\begin{proof}
Let
\[
\mathbf{s}_t := \mathbf{B}_t \mathbf{x}_t - \mathbf{q}_t,
\qquad
\mathbf{g}_t(\boldsymbol{\theta}) := \nabla_{\mathbf{x}} c_t(\mathbf{x}_t;\boldsymbol{\theta}),
\]
and recall the structural KKT diagnostic
\[
\mathcal{L}^{\mathrm{KKT}}_t(\boldsymbol{\theta})
:=
\inf_{\boldsymbol{\lambda}\ge 0}
\Big(
\|(\mathbf{s}_t)_+\|^2
+
\|\mathbf{g}_t(\boldsymbol{\theta}) + \mathbf{B}_t^\top \boldsymbol{\lambda}\|^2
+
\big|
\boldsymbol{\lambda}^\top \mathbf{s}_t
\big|
\Big).
\]

Both directions are established below.

\medskip
\noindent
\textbf{(Forward implication).}
Assume that $\mathcal{L}^{\mathrm{KKT}}_t(\boldsymbol{\theta})=0$. Since each term inside the infimum is nonnegative, 
the first term must vanish:
\[
\|(\mathbf{s}_t)_+\|^2 = 0.
\]
Hence $(\mathbf{s}_t)_i \le 0$ for all $i$, that is,
\[
\mathbf{s}_t = \mathbf{B}_t \mathbf{x}_t - \mathbf{q}_t \le 0,
\]
so primal feasibility holds:
\[
\mathbf{B}_t \mathbf{x}_t \le \mathbf{q}_t.
\]

Next, because the infimum is zero, there exists a sequence 
$\{\boldsymbol{\lambda}^{(m)}\}_{m\ge 1}$ with $\boldsymbol{\lambda}^{(m)} \ge 0$ such that
\begin{equation}
\label{eq:kkt-seq-1}
\|\mathbf{g}_t(\boldsymbol{\theta}) + \mathbf{B}_t^\top \boldsymbol{\lambda}^{(m)}\| \to 0,
\qquad
\big|(\boldsymbol{\lambda}^{(m)})^\top \mathbf{s}_t\big| \to 0.
\end{equation}

Since $\mathbf{s}_t \le 0$ and $\boldsymbol{\lambda}^{(m)} \ge 0$, each scalar
$(\boldsymbol{\lambda}^{(m)})^\top \mathbf{s}_t \le 0$. Therefore
\[
(\boldsymbol{\lambda}^{(m)})^\top \mathbf{s}_t \to 0.
\]

Let the active and inactive sets be
\[
\mathcal{A}_t := \{i : (\mathbf{s}_t)_i = 0\},
\qquad
\mathcal{I}_t := \{i : (\mathbf{s}_t)_i < 0\}.
\]
Because every summand $\lambda^{(m)}_i (\mathbf{s}_t)_i$ is nonpositive and their sum converges to zero, 
it follows that
\[
\lambda^{(m)}_i \to 0
\qquad \text{for every } i \in \mathcal{I}_t.
\]
Hence the inactive multipliers vanish asymptotically.

Write $\mathbf{B}_{t,\mathcal{A}}$ for the submatrix formed by the active constraints. Then
\[
\mathbf{g}_t(\boldsymbol{\theta}) + \mathbf{B}_{t,\mathcal{A}}^\top \boldsymbol{\lambda}^{(m)}_{\mathcal{A}}
\to 0,
\]
because the contribution of inactive multipliers converges to zero. Therefore
\[
-\mathbf{g}_t(\boldsymbol{\theta})
\in
\overline{
\left\{
\mathbf{B}_{t,\mathcal{A}}^\top \boldsymbol{\lambda}_{\mathcal{A}}
:
\boldsymbol{\lambda}_{\mathcal{A}} \ge 0
\right\}
}.
\]

Now define
\[
\mathcal{K}_t :=
\left\{
\mathbf{B}_{t,\mathcal{A}}^\top \boldsymbol{\lambda}_{\mathcal{A}}
:
\boldsymbol{\lambda}_{\mathcal{A}} \ge 0
\right\}.
\]
Since $\mathbf{B}_{t,\mathcal{A}}$ has finitely many columns, $\mathcal{K}_t$ is a finitely generated 
polyhedral cone, hence closed. It follows that
\[
-\mathbf{g}_t(\boldsymbol{\theta}) \in \mathcal{K}_t.
\]
Therefore, there exists $\bar{\boldsymbol{\lambda}}_{\mathcal{A}} \ge 0$ such that
\[
\mathbf{g}_t(\boldsymbol{\theta}) + \mathbf{B}_{t,\mathcal{A}}^\top \bar{\boldsymbol{\lambda}}_{\mathcal{A}} = 0.
\]

Extend this vector to $\bar{\boldsymbol{\lambda}}_t \in \mathbb{R}^k_+$ by setting
\[
(\bar{\boldsymbol{\lambda}}_t)_i :=
\begin{cases}
(\bar{\boldsymbol{\lambda}}_{\mathcal{A}})_i, & i \in \mathcal{A}_t,\\
0, & i \in \mathcal{I}_t.
\end{cases}
\]
Then $\bar{\boldsymbol{\lambda}}_t \ge 0$ and
\[
\mathbf{g}_t(\boldsymbol{\theta}) + \mathbf{B}_t^\top \bar{\boldsymbol{\lambda}}_t = 0.
\]
Moreover, because $\bar{\boldsymbol{\lambda}}_t$ vanishes on inactive constraints and 
$\mathbf{s}_t=0$ on active constraints, it follows that
\[
\bar{\boldsymbol{\lambda}}_t^\top \mathbf{s}_t = 0.
\]

Thus there exists $\bar{\boldsymbol{\lambda}}_t \ge 0$ such that
\[
\mathbf{g}_t(\boldsymbol{\theta}) + \mathbf{B}_t^\top \bar{\boldsymbol{\lambda}}_t = 0,
\qquad
\mathbf{B}_t \mathbf{x}_t \le \mathbf{q}_t,
\qquad
\bar{\boldsymbol{\lambda}}_t^\top(\mathbf{B}_t \mathbf{x}_t - \mathbf{q}_t)=0,
\]
which is exactly the KKT system at $(\mathbf{x}_t,\boldsymbol{\theta})$.

\medskip
\noindent
\textbf{(Reverse implication).}
Conversely, suppose there exists $\bar{\boldsymbol{\lambda}}_t \ge 0$ such that
\[
\mathbf{g}_t(\boldsymbol{\theta}) + \mathbf{B}_t^\top \bar{\boldsymbol{\lambda}}_t = 0,
\qquad
\mathbf{B}_t \mathbf{x}_t \le \mathbf{q}_t,
\qquad
\bar{\boldsymbol{\lambda}}_t^\top(\mathbf{B}_t \mathbf{x}_t - \mathbf{q}_t)=0.
\]
Then each term in the definition of $\mathcal{L}^{\mathrm{KKT}}_t(\boldsymbol{\theta})$ vanishes at 
$\boldsymbol{\lambda}=\bar{\boldsymbol{\lambda}}_t$:
\[
\|(\mathbf{s}_t)_+\|^2 = 0,
\qquad
\|\mathbf{g}_t(\boldsymbol{\theta}) + \mathbf{B}_t^\top \bar{\boldsymbol{\lambda}}_t\|^2 = 0,
\qquad
\big|
\bar{\boldsymbol{\lambda}}_t^\top \mathbf{s}_t
\big| = 0.
\]
Therefore,
\[
\mathcal{L}^{\mathrm{KKT}}_t(\boldsymbol{\theta}) = 0.
\]

This proves the equivalence.
\end{proof}

\subsection{Proof of Lemma~3 (Projected-gradient consistency)}
\label{app:proof:proj-consistency-lemma}

\begin{proof}
Assume that $\mathcal{L}^{\mathrm{KKT}}_t(\boldsymbol{\theta}) = 0$. By Lemma~2, 
there exists a multiplier vector $\bar{\boldsymbol{\lambda}}_t \ge 0$ such that the 
KKT conditions hold at $(\mathbf{x}_t,\boldsymbol{\theta})$, namely
\[
\nabla_{\mathbf{x}} c_t(\mathbf{x}_t;\boldsymbol{\theta})
+
\mathbf{B}_t^\top \bar{\boldsymbol{\lambda}}_t
=
0.
\]

Applying the linear operator $\Proj_t$ to both sides yields
\[
\Proj_t \nabla_{\mathbf{x}} c_t(\mathbf{x}_t;\boldsymbol{\theta})
+
\Proj_t \mathbf{B}_t^\top \bar{\boldsymbol{\lambda}}_t
=
0.
\]

By Lemma~1, $\Proj_t$ is the orthogonal projector onto 
$\mathsf{ker}(\mathbf{B}_t)$ and annihilates $\mathsf{range}(\mathbf{B}_t^\top)$, i.e.,
\[
\Proj_t \mathbf{B}_t^\top = 0.
\]
Since $\mathbf{B}_t^\top \bar{\boldsymbol{\lambda}}_t \in \mathsf{range}(\mathbf{B}_t^\top)$,
it follows that
\[
\Proj_t \mathbf{B}_t^\top \bar{\boldsymbol{\lambda}}_t = 0.
\]

Therefore,
\[
\Proj_t \nabla_{\mathbf{x}} c_t(\mathbf{x}_t;\boldsymbol{\theta}) = 0.
\]
\end{proof}
This establishes the projected-gradient consistency condition. The result shows that when the structural KKT diagnostic vanishes, the gradient lies entirely 
in $\mathsf{range}(\mathbf{B}_t^\top)$, i.e., in the normal space induced by the 
constraints. Equivalently, the component of the gradient along feasible directions, 
represented by $\mathsf{ker}(\mathbf{B}_t)$, must vanish. This property is crucial for identifiability: it provides a constraint on 
$\boldsymbol{\theta}$ that can be exploited to recover the parameter uniquely 
under appropriate structural conditions.

\subsection{Proof of Theorem~1 (Local identifiability)}
\label{app:proof:ident}

\begin{proof}
Fix a period $t$ and consider a parameter $\boldsymbol{\theta}$ in a neighborhood 
of the true parameter $\boldsymbol{\theta}_t$ such that 
Assumption~2 holds locally.

Assume that
\[
\mathcal{L}^{\mathrm{KKT}}_t(\boldsymbol{\theta}) = 0.
\]

By Lemma~3, zero structural KKT diagnostic implies
\begin{equation}
\label{eq:ident-candidate}
\Proj_t \nabla_{\mathbf{x}} c_t(\mathbf{x}_t;\boldsymbol{\theta}) = 0.
\end{equation}

The corresponding condition for the true parameter is obtained next. 
$\boldsymbol{\theta}_t$. Under Assumption~1, the observed decision 
$\mathbf{x}_t$ is consistent with an optimal solution under 
$\boldsymbol{\theta}_t$, so the KKT conditions hold at 
$(\mathbf{x}_t,\boldsymbol{\theta}_t)$. Hence, there exists 
$\boldsymbol{\lambda}_t^\star \ge 0$ such that
\[
\nabla_{\mathbf{x}} c_t(\mathbf{x}_t;\boldsymbol{\theta}_t)
+
\mathbf{B}_t^\top \boldsymbol{\lambda}_t^\star
=
0.
\]

Applying the linear operator $\Proj_t$ and using 
$\Proj_t \mathbf{B}_t^\top = 0$ from Lemma~1, it follows that
\begin{equation}
\label{eq:ident-true}
\Proj_t \nabla_{\mathbf{x}} c_t(\mathbf{x}_t;\boldsymbol{\theta}_t) = 0.
\end{equation}

Combining Equations~\eqref{eq:ident-candidate} and~\eqref{eq:ident-true} yields
\[
\Proj_t \nabla_{\mathbf{x}} c_t(\mathbf{x}_t;\boldsymbol{\theta})
=
\Proj_t \nabla_{\mathbf{x}} c_t(\mathbf{x}_t;\boldsymbol{\theta}_t).
\]

Define the mapping
\[
\Phi_t(\boldsymbol{\vartheta})
:=
\Proj_t \nabla_{\mathbf{x}} c_t(\mathbf{x}_t;\boldsymbol{\vartheta}).
\]

By Assumption~2, the mapping $\Phi_t$ is locally injective in a 
neighborhood of $\boldsymbol{\theta}_t$. Since both $\boldsymbol{\theta}$ and 
$\boldsymbol{\theta}_t$ lie in that neighborhood and satisfy
\[
\Phi_t(\boldsymbol{\theta}) = \Phi_t(\boldsymbol{\theta}_t),
\]
local injectivity implies
\[
\boldsymbol{\theta} = \boldsymbol{\theta}_t.
\]

Therefore, for every $\boldsymbol{\theta}$ sufficiently close to 
$\boldsymbol{\theta}_t$,
\[
\mathcal{L}^{\mathrm{KKT}}_t(\boldsymbol{\theta}) = 0
\quad \Longrightarrow \quad
\boldsymbol{\theta} = \boldsymbol{\theta}_t,
\]
which establishes local identifiability. 
\end{proof}
The result is stated for a fixed period $t$ and establishes local identifiability 
at that period. In dynamic settings where $\boldsymbol{\theta}_t$ evolves over time, 
this property applies pointwise along the trajectory 
$\{\boldsymbol{\theta}_t\}_{t=1}^T$, forming the basis for trajectory-level 
recovery and dynamic tracking guarantees developed later.

\subsection{Proof of Corollary~1 (Stationary case)}
\label{app:proof:cor-stationary}

\begin{proof}
Assume that the latent parameter is stationary, i.e.,
\[
\boldsymbol{\theta}_t \equiv \boldsymbol{\theta}
\qquad \text{for all } t=1,\dots,T.
\]

Suppose that
\[
\sum_{t=1}^T \mathcal{L}^{\mathrm{KKT}}_t(\boldsymbol{\theta}') = 0.
\]
Since each term $\mathcal{L}^{\mathrm{KKT}}_t(\boldsymbol{\theta}')$ is nonnegative by definition, it follows that
\[
\mathcal{L}^{\mathrm{KKT}}_t(\boldsymbol{\theta}') = 0
\qquad \text{for every } t=1,\dots,T.
\]

Fix any period $t$. By Theorem~1, there exists a neighborhood 
$\mathcal{U}_t$ of $\boldsymbol{\theta}_t$ such that, for all 
$\boldsymbol{\vartheta} \in \mathcal{U}_t$,
\[
\mathcal{L}^{\mathrm{KKT}}_t(\boldsymbol{\vartheta}) = 0
\quad \Longrightarrow \quad
\boldsymbol{\vartheta} = \boldsymbol{\theta}_t.
\]

Assume that $\boldsymbol{\theta}'$ lies in the intersection of these neighborhoods,
\[
\boldsymbol{\theta}' \in \bigcap_{t=1}^T \mathcal{U}_t.
\]

Then, for each $t$, the condition $\mathcal{L}^{\mathrm{KKT}}_t(\boldsymbol{\theta}') = 0$ implies
\[
\boldsymbol{\theta}' = \boldsymbol{\theta}_t.
\]

Because the parameter is stationary, one has $\boldsymbol{\theta}_t = \boldsymbol{\theta}$ 
for all $t$, and hence
\[
\boldsymbol{\theta}' = \boldsymbol{\theta}.
\]

Therefore,
\[
\sum_{t=1}^T \mathcal{L}^{\mathrm{KKT}}_t(\boldsymbol{\theta}') = 0
\quad \Longrightarrow \quad
\boldsymbol{\theta}' = \boldsymbol{\theta},
\]
which proves the claim.
\end{proof}

\medskip
\noindent
\textbf{Interpretation.}
The condition $\sum_{t=1}^T \mathcal{L}^{\mathrm{KKT}}_t(\boldsymbol{\theta}')=0$ corresponds to a noiseless 
and perfectly specified setting in which the observed decisions exactly satisfy the 
KKT conditions under a single parameter. Such exact recovery is not expected in 
practical settings with noise, model misspecification, or time-varying parameters. 
Accordingly, this result should be interpreted as an idealized benchmark that 
characterizes identifiability in the stationary case, and serves as a reference point 
for the approximate recovery and trajectory-tracking guarantees developed under 
noise and drift in subsequent sections.

\subsection{Proof of Proposition~1 (Consistency implies stationarity)}
\label{app:proof:prop-consistency}

\begin{proof}
Suppose there exists a parameter $\bar{\boldsymbol{\theta}} \in \Theta$ such that
\[
\mathcal{L}^{\mathrm{KKT}}_t(\bar{\boldsymbol{\theta}})=0
\qquad \text{for all } t=1,\dots,T.
\]

Fix any period $t$. Since $\mathcal{L}^{\mathrm{KKT}}_t(\bar{\boldsymbol{\theta}})=0$, 
Theorem~1 implies that, provided $\bar{\boldsymbol{\theta}}$ lies 
in a neighborhood of $\boldsymbol{\theta}_t$ where Assumption~2 applies,
\[
\bar{\boldsymbol{\theta}} = \boldsymbol{\theta}_t.
\]

Because this argument holds for every $t$, it follows that
\[
\boldsymbol{\theta}_1 = \boldsymbol{\theta}_2 = \cdots = \boldsymbol{\theta}_T
= \bar{\boldsymbol{\theta}}.
\]

Therefore, the parameter sequence is constant over time:
\[
\boldsymbol{\theta}_t \equiv \bar{\boldsymbol{\theta}}
\qquad \text{for all } t=1,\dots,T.
\]

This establishes the claim.
\end{proof}
% \textbf{Interpretation.}
The result characterizes an idealized regime in which a single parameter 
perfectly explains all observed decisions across time. Under the local 
identifiability condition, such global consistency is only possible if the 
underlying parameter sequence is stationary.
% \textbf{Remark (Dynamic compatibility).}
This proposition does not contradict the presence of time-varying parameters. 
Rather, it shows that the existence of a common parameter achieving zero 
structural KKT diagnostic across all periods corresponds to a degenerate case of the 
dynamic model. In realistic settings with noise or drift, exact consistency 
fails, which motivates the trajectory-based estimation and regret analysis 
developed in subsequent sections.

\subsection{Proof of Theorem~2 (Linear identifiability)}
\label{app:proof:linear-ident}

\begin{proof}
Assume that the gradient of the cost function is affine in the parameter:
\[
\nabla_{\mathbf{x}} c_t(\mathbf{x};\boldsymbol{\theta})
=
\mathbf{A}_t(\mathbf{x})\boldsymbol{\theta} + \mathbf{b}_t(\mathbf{x}).
\]

Fix a period $t$ and evaluate at the observed decision $\mathbf{x}_t$. Define
\[
\widetilde{\mathbf{A}}_t := \Proj_t \mathbf{A}_t(\mathbf{x}_t),
\qquad
\widetilde{\mathbf{b}}_t := \Proj_t \mathbf{b}_t(\mathbf{x}_t).
\]
Then
\[
\Proj_t \nabla_{\mathbf{x}} c_t(\mathbf{x}_t;\boldsymbol{\theta})
=
\widetilde{\mathbf{A}}_t \boldsymbol{\theta} + \widetilde{\mathbf{b}}_t.
\]

\medskip
\noindent
\textbf{Part (i): Pointwise identifiability.}

Suppose that
\[
\mathcal{L}^{\mathrm{KKT}}_t(\boldsymbol{\theta}) = 0.
\]
By Lemma~3,
\[
\Proj_t \nabla_{\mathbf{x}} c_t(\mathbf{x}_t;\boldsymbol{\theta}) = 0.
\]
By the linear structure of the gradient, this implies
\[
\widetilde{\mathbf{A}}_t \boldsymbol{\theta} + \widetilde{\mathbf{b}}_t = 0.
\]

Similarly, for the true parameter $\boldsymbol{\theta}_t$,
\[
\widetilde{\mathbf{A}}_t \boldsymbol{\theta}_t + \widetilde{\mathbf{b}}_t = 0.
\]
Subtracting yields
\[
\widetilde{\mathbf{A}}_t(\boldsymbol{\theta} - \boldsymbol{\theta}_t) = 0.
\]

If $\widetilde{\mathbf{A}}_t$ has full column rank, then its null space is trivial, 
which implies
\[
\boldsymbol{\theta} = \boldsymbol{\theta}_t.
\]

\medskip
\noindent
\textbf{Part (ii): Stationary/global recovery.}

Now assume the stationary setting
\[
\boldsymbol{\theta}_t \equiv \boldsymbol{\theta}
\qquad \text{for all } t=1,\dots,T.
\]

Define the stacked projected Jacobian
\[
\mathcal{A}
:=
\begin{bmatrix}
\widetilde{\mathbf{A}}_1\\
\widetilde{\mathbf{A}}_2\\
\vdots\\
\widetilde{\mathbf{A}}_T
\end{bmatrix}.
\]

Suppose that a candidate parameter $\boldsymbol{\theta}'$ satisfies
\[
\sum_{t=1}^T \mathcal{L}^{\mathrm{KKT}}_t(\boldsymbol{\theta}') = 0.
\]
Then each $\mathcal{L}^{\mathrm{KKT}}_t(\boldsymbol{\theta}')=0$, and by Lemma~3,
\[
\widetilde{\mathbf{A}}_t \boldsymbol{\theta}' + \widetilde{\mathbf{b}}_t = 0
\quad \text{for all } t.
\]
Similarly,
\[
\widetilde{\mathbf{A}}_t \boldsymbol{\theta} + \widetilde{\mathbf{b}}_t = 0
\quad \text{for all } t.
\]

Subtracting and stacking yields
\[
\mathcal{A}(\boldsymbol{\theta}' - \boldsymbol{\theta}) = 0.
\]

If $\mathcal{A}$ has full column rank, then
\[
\boldsymbol{\theta}' = \boldsymbol{\theta}.
\]
\end{proof}
% \textbf{Interpretation.}
The theorem replaces the abstract local injectivity condition with an explicit 
rank condition. Pointwise identifiability depends on the rank of the projected 
Jacobian at a single period, while stationary/global recovery depends on the 
rank of the stacked projected Jacobian across periods. This provides a concrete 
and verifiable characterization of identifiability in the linear setting.
% \textbf{Remark (Dynamic compatibility).}
The stationary component serves as an idealized benchmark and does not restrict 
the presence of time-varying parameters. In dynamic settings, the failure of a 
single parameter to satisfy these conditions motivates trajectory-based estimation 
and regret analysis developed in subsequent sections.

\subsection{Proof of Corollary~2 (Rank-based stability modulus)}
\label{app:proof:rank-modulus}

\begin{proof}
The proof begins with the pointwise case. Recall from the proof of 
Theorem~2 that, under the linear gradient structure,
\[
\Proj_t \nabla_{\mathbf{x}} c_t(\mathbf{x}_t;\boldsymbol{\theta})
=
\widetilde{\mathbf{A}}_t \boldsymbol{\theta} + \widetilde{\mathbf{b}}_t,
\]
where
\[
\widetilde{\mathbf{A}}_t := \Proj_t \mathbf{A}_t(\mathbf{x}_t),
\qquad
\widetilde{\mathbf{b}}_t := \Proj_t \mathbf{b}_t(\mathbf{x}_t).
\]

For any $\boldsymbol{\theta}$ and the true parameter $\boldsymbol{\theta}_t$, one has
\[
\Proj_t \nabla_{\mathbf{x}} c_t(\mathbf{x}_t;\boldsymbol{\theta})
-
\Proj_t \nabla_{\mathbf{x}} c_t(\mathbf{x}_t;\boldsymbol{\theta}_t)
=
\widetilde{\mathbf{A}}_t(\boldsymbol{\theta} - \boldsymbol{\theta}_t).
\]
Taking norms yields
\[
\big\|
\Proj_t \nabla_{\mathbf{x}} c_t(\mathbf{x}_t;\boldsymbol{\theta})
-
\Proj_t \nabla_{\mathbf{x}} c_t(\mathbf{x}_t;\boldsymbol{\theta}_t)
\big\|^2
=
\big\|
\widetilde{\mathbf{A}}_t(\boldsymbol{\theta} - \boldsymbol{\theta}_t)
\big\|^2.
\]

If $\widetilde{\mathbf{A}}_t$ has full column rank, its smallest singular value 
$\sigma_{\min}(\widetilde{\mathbf{A}}_t)$ is strictly positive. Hence,
\[
\big\|
\widetilde{\mathbf{A}}_t(\boldsymbol{\theta} - \boldsymbol{\theta}_t)
\big\|^2
\;\ge\;
\sigma_{\min}^2(\widetilde{\mathbf{A}}_t)\,
\|\boldsymbol{\theta} - \boldsymbol{\theta}_t\|^2.
\]

Therefore,
\begin{equation}
\label{eq:rank-stability-pointwise-final}
\big\|
\Proj_t \nabla_{\mathbf{x}} c_t(\mathbf{x}_t;\boldsymbol{\theta})
-
\Proj_t \nabla_{\mathbf{x}} c_t(\mathbf{x}_t;\boldsymbol{\theta}_t)
\big\|^2
\;\ge\;
\sigma_{\min}^2(\widetilde{\mathbf{A}}_t)\,
\|\boldsymbol{\theta} - \boldsymbol{\theta}_t\|^2.
\end{equation}

\medskip
\noindent
For the stationary/global case, define
\[
\mathcal{A}
:=
\begin{bmatrix}
\widetilde{\mathbf{A}}_1\\
\vdots\\
\widetilde{\mathbf{A}}_T
\end{bmatrix}.
\]
Then
\[
\|\mathcal{A}(\boldsymbol{\theta}' - \boldsymbol{\theta})\|^2
\;\ge\;
\sigma_{\min}^2(\mathcal{A})\,
\|\boldsymbol{\theta}' - \boldsymbol{\theta}\|^2.
\]
\end{proof}
% \textbf{Interpretation.}
The singular-value bound quantifies the strength of identifiability: larger 
$\sigma_{\min}$ implies greater robustness to perturbations. In particular, the 
projected gradient mapping is not only injective but also locally stable, with 
modulus determined by $\sigma_{\min}^2$.
% \textbf{Remark (Connection to structural KKT diagnostic).}
By Lemma~3, zero structural KKT diagnostic enforces 
$\Proj_t \nabla_{\mathbf{x}} c_t(\mathbf{x}_t;\boldsymbol{\theta}) = 0$. 
More generally, deviations from zero projected gradient quantify violations of 
KKT consistency, so inequality~\eqref{eq:rank-stability-pointwise-final} provides 
a lower bound linking parameter error to structural KKT diagnostic deviations.
% \textbf{Remark (Robustness under noise).}
In noisy or approximate settings, the stability modulus determines how estimation 
error scales with residual loss. This relationship is formalized in the noise 
robustness results of Section~4.6.

\subsection{Proof of Lemma~4 (Identifiability vs stability gap)}
\label{app:proof:identifiability-gap}

\begin{proof}
The argument shows that local identifiability at each period does not by itself imply
stable recovery of the full parameter sequence under perturbations.

Fix a period $t$ and define the projected-gradient mapping
\[
\Phi_t(\boldsymbol{\theta})
:=
\Proj_t \nabla_{\mathbf{x}} c_t(\mathbf{x}_t;\boldsymbol{\theta}).
\]
By Assumption~2, $\Phi_t$ is locally injective in a neighborhood
of the true parameter $\boldsymbol{\theta}_t$. Therefore, in the noiseless
benchmark, the equality
\[
\Phi_t(\boldsymbol{\theta})=\Phi_t(\boldsymbol{\theta}_t)
\]
implies $\boldsymbol{\theta}=\boldsymbol{\theta}_t$ locally. This is a
uniqueness statement, but it does not yet control the sensitivity of the
inverse map to perturbations.

To study sensitivity, consider the linear-gradient setting of
Theorem~2, where
\[
\nabla_{\mathbf{x}} c_t(\mathbf{x}_t;\boldsymbol{\theta})
=
\mathbf{A}_t(\mathbf{x}_t)\boldsymbol{\theta}
+
\mathbf{b}_t(\mathbf{x}_t).
\]
Applying the projector $\Proj_t$ yields
\[
\Phi_t(\boldsymbol{\theta})
=
\Proj_t \mathbf{A}_t(\mathbf{x}_t)\boldsymbol{\theta}
+
\Proj_t \mathbf{b}_t(\mathbf{x}_t)
=
\widetilde{\mathbf{A}}_t \boldsymbol{\theta}
+
\widetilde{\mathbf{b}}_t,
\]
where
\[
\widetilde{\mathbf{A}}_t := \Proj_t \mathbf{A}_t(\mathbf{x}_t),
\qquad
\widetilde{\mathbf{b}}_t := \Proj_t \mathbf{b}_t(\mathbf{x}_t).
\]

Now let $\widehat{\boldsymbol{\theta}}_t$ denote a perturbed per-period
estimate satisfying
\[
\Phi_t(\widehat{\boldsymbol{\theta}}_t)
=
\Phi_t(\boldsymbol{\theta}_t)+\mathbf{e}_t,
\]
where $\mathbf{e}_t$ is a residual perturbation induced by noise,
approximation error, or numerical inversion error. Substituting the affine
representation above gives
\[
\widetilde{\mathbf{A}}_t \widehat{\boldsymbol{\theta}}_t
+
\widetilde{\mathbf{b}}_t
=
\widetilde{\mathbf{A}}_t \boldsymbol{\theta}_t
+
\widetilde{\mathbf{b}}_t
+
\mathbf{e}_t,
\]
and hence
\[
\widetilde{\mathbf{A}}_t
(\widehat{\boldsymbol{\theta}}_t-\boldsymbol{\theta}_t)
=
\mathbf{e}_t.
\]

If $\widetilde{\mathbf{A}}_t$ has full column rank, then
Corollary~2 yields the perturbation bound
\[
\|\widehat{\boldsymbol{\theta}}_t-\boldsymbol{\theta}_t\|
\le
\frac{1}{\sigma_{\min}(\widetilde{\mathbf{A}}_t)}
\|\mathbf{e}_t\|.
\]
This inequality shows that inverse sensitivity is amplified by the factor
$1/\sigma_{\min}(\widetilde{\mathbf{A}}_t)$. Therefore, even when local
identifiability holds, the inverse problem may be ill-conditioned if
$\sigma_{\min}(\widetilde{\mathbf{A}}_t)$ is small. In that case, small
perturbations in the projected residual can induce large deviations in the
recovered parameter.

Now consider recovery performed independently at each period. Since no
temporal coupling is imposed, each estimate
$\widehat{\boldsymbol{\theta}}_t$ is determined only by its own local
perturbation $\mathbf{e}_t$. Consequently, the recovered path
$\{\widehat{\boldsymbol{\theta}}_t\}_{t=1}^T$ may exhibit large temporal
variation:
\[
\sum_{t=2}^T
\|\widehat{\boldsymbol{\theta}}_t-\widehat{\boldsymbol{\theta}}_{t-1}\|.
\]
Indeed, writing
\[
\widehat{\boldsymbol{\theta}}_t-\widehat{\boldsymbol{\theta}}_{t-1}
=
(\boldsymbol{\theta}_t-\boldsymbol{\theta}_{t-1})
+
(\widehat{\boldsymbol{\theta}}_t-\boldsymbol{\theta}_t)
-
(\widehat{\boldsymbol{\theta}}_{t-1}-\boldsymbol{\theta}_{t-1}),
\]
and applying the triangle inequality gives
\[
\|\widehat{\boldsymbol{\theta}}_t-\widehat{\boldsymbol{\theta}}_{t-1}\|
\le
\|\boldsymbol{\theta}_t-\boldsymbol{\theta}_{t-1}\|
+
\|\widehat{\boldsymbol{\theta}}_t-\boldsymbol{\theta}_t\|
+
\|\widehat{\boldsymbol{\theta}}_{t-1}-\boldsymbol{\theta}_{t-1}\|.
\]
Using the perturbation bound above,
\[
\|\widehat{\boldsymbol{\theta}}_t-\widehat{\boldsymbol{\theta}}_{t-1}\|
\le
\|\boldsymbol{\theta}_t-\boldsymbol{\theta}_{t-1}\|
+
\frac{\|\mathbf{e}_t\|}{\sigma_{\min}(\widetilde{\mathbf{A}}_t)}
+
\frac{\|\mathbf{e}_{t-1}\|}{\sigma_{\min}(\widetilde{\mathbf{A}}_{t-1})}.
\]
Summing over $t=2,\dots,T$ yields
\[
\sum_{t=2}^T
\|\widehat{\boldsymbol{\theta}}_t-\widehat{\boldsymbol{\theta}}_{t-1}\|
\le
V_T
+
\sum_{t=2}^T \frac{\|\mathbf{e}_t\|}{\sigma_{\min}(\widetilde{\mathbf{A}}_t)}
+
\sum_{t=2}^T \frac{\|\mathbf{e}_{t-1}\|}{\sigma_{\min}(\widetilde{\mathbf{A}}_{t-1})}.
\]
Hence the total variation of the recovered sequence can be substantially
larger than the true variation budget $V_T$ whenever the residual
perturbations are non-negligible or the projected Jacobians are poorly
conditioned. In particular, local identifiability alone provides no
mechanism for controlling the temporal roughness of independently recovered
per-period estimates.

By contrast, the trajectory-based estimator imposes explicit temporal
coupling through the regularization term
\[
\lambda \sum_{t=2}^T d(\boldsymbol{\theta}_t,\boldsymbol{\theta}_{t-1}),
\]
which penalizes abrupt intertemporal changes and suppresses
perturbation-induced oscillations. This coupling stabilizes the recovered
path by controlling variation across periods, rather than treating each
period in isolation.

Therefore, local identifiability provides a periodwise uniqueness property under the stated sufficient conditions, but it does not imply stability of the recovered parameter sequence under perturbations. Independent per-period estimation may remain unstable, whereas trajectory-based estimation can regularize such instability through temporal coupling.
\end{proof}

\subsection{Proof of Theorem~3 (Existence and well-posedness)}
\label{app:proof:exist-unique}

\begin{proof}
Existence and well-posedness are established separately for the forward problem and the inverse trajectory estimator.

\medskip
\noindent
\textbf{(i) Forward problem: existence and uniqueness.}

Fix any period $t$. Consider the forward problem
\[
\min_{\mathbf{x}\in \mathcal{X}_t} c_t(\mathbf{x};\boldsymbol{\theta}_t).
\]

Under Assumption~3, either $\mathcal{X}_t$ is compact or the objective 
$c_t(\cdot;\boldsymbol{\theta}_t)$ is level-bounded on $\mathcal{X}_t$. Since 
$c_t(\cdot;\boldsymbol{\theta}_t)$ is continuous, the Weierstrass theorem implies 
that a minimizer $\mathbf{x}_t^\star$ exists.

To establish uniqueness, suppose that $\mathbf{x}_t^{(1)} \neq \mathbf{x}_t^{(2)}$ are both minimizers. 
By Assumption~4 (quadratic growth), for any $\mathbf{x}\in\mathcal{X}_t$,
\[
c_t(\mathbf{x};\boldsymbol{\theta}_t) - c_t^\star(\boldsymbol{\theta}_t)
\;\ge\;
\frac{\alpha_t}{2}\|\mathbf{x}-\mathbf{x}_t^\star\|^2,
\]
for some $\alpha_t>0$. Evaluating this inequality at a minimizer distinct from $\mathbf{x}_t^\star$ 
yields a strictly positive value, contradicting optimality. Hence the minimizer is unique.

\medskip
\noindent
\textbf{(ii) Inverse trajectory estimator: existence.}

Consider the trajectory estimator
\[
\{\hat{\boldsymbol{\theta}}_t\}_{t=1}^T
\in
\arg\min_{\boldsymbol{\theta}_1,\dots,\boldsymbol{\theta}_T\in\Theta}
\left\{
\sum_{t=1}^T \mathcal{L}^{\mathrm{KKT}}_t(\boldsymbol{\theta}_t)
+
\lambda \sum_{t=2}^T d(\boldsymbol{\theta}_t,\boldsymbol{\theta}_{t-1})
\right\}.
\]

Define
\[
\mathcal{J}_T(\boldsymbol{\theta}_{1:T})
:=
\sum_{t=1}^T \mathcal{L}^{\mathrm{KKT}}_t(\boldsymbol{\theta}_t)
+
\lambda \sum_{t=2}^T d(\boldsymbol{\theta}_t,\boldsymbol{\theta}_{t-1}),
\quad
\boldsymbol{\theta}_{1:T}\in \Theta^T.
\]

By Assumption~5, $\Theta$ is compact and each $\mathcal{L}^{\mathrm{KKT}}_t$ is lower semicontinuous in $\boldsymbol{\theta}_t$. The temporal regularization term is also lower semicontinuous. Therefore, $\mathcal{J}_T$ is lower semicontinuous on the compact product space $\Theta^T$. By the Weierstrass theorem, $\mathcal{J}_T$ attains its minimum, so the inverse trajectory estimator admits at least one minimizer.
\end{proof}

% \medskip
% \noindent
% \textbf{Discussion (Link to identifiability).}
The identifiability results of Section~4.3 characterize the 
structure of the zero-loss set. In particular, in the noiseless limit, if a parameter 
$\boldsymbol{\vartheta}_t$ satisfies $\mathcal{L}^{\mathrm{KKT}}_t(\boldsymbol{\vartheta}_t)=0$, then 
Theorem~1 implies locally that
\[
\boldsymbol{\vartheta}_t = \boldsymbol{\theta}_t.
\]
In the stationary benchmark, Corollary~1 and 
Proposition~1 show that a single parameter 
can achieve zero loss across all periods only when the system is stationary. 
Thus, identifiability ensures that zero-loss solutions correspond to the true 
parameters in idealized settings.

% \noindent
% \textbf{Remark (Stability and multiplicity).}
The trajectory estimator may admit multiple minimizers in general due to noise and model misspecification. The regularization term discourages excessive path variation, while the identifiability and local stability results in Section~4.3 characterize recoverability under the conditions stated in the main manuscript.

% \noindent
% \textbf{Remark (Regularization structure).}
The temporal regularization $d(\cdot,\cdot)$ is assumed to be lower semicontinuous and compatible with the topology of $\Theta$ (e.g., a metric or a Bregman divergence), which preserves lower semicontinuity of the trajectory objective.

    % \noindent
    % \textbf{Remark (Well-posedness interpretation).}
The existence result guarantees attainment of a minimizer under the stated compactness and lower-semicontinuity conditions. Uniqueness of the inverse trajectory is not guaranteed in general. Local sensitivity to residual perturbations is characterized separately by the stability results in Sections~4.3 and~4.6.

\subsection{Proof of Theorem~4 (Static regret)}
\label{app:proof:static-regret}

\begin{proof}
Let $\{\boldsymbol{\theta}_t\}_{t=1}^T$ be generated by projected online gradient descent (OGD):
\[
\boldsymbol{\theta}_{t+1}
=
\Pi_{\Theta}\bigl(\boldsymbol{\theta}_t - \eta_t \nabla \ell_t(\boldsymbol{\theta}_t)\bigr),
\]
where $\Pi_{\Theta}$ denotes Euclidean projection onto $\Theta$.

Assume that each $\ell_t$ is convex and has uniformly bounded subgradients:
\[
\|\nabla \ell_t(\boldsymbol{\theta})\| \le G \quad \forall \theta \in \Theta, \forall t.
\]
Assume also that $\Theta$ has diameter $D$.

\medskip
\noindent
\textbf{Step 1: One-step inequality.}

By non-expansiveness of projection,
\[
\|\boldsymbol{\theta}_{t+1} - \boldsymbol{\theta}^\star\|^2
\le
\|\boldsymbol{\theta}_t - \eta_t \nabla \ell_t(\boldsymbol{\theta}_t) - \boldsymbol{\theta}^\star\|^2.
\]
Expanding:
\[
= \|\boldsymbol{\theta}_t - \boldsymbol{\theta}^\star\|^2
-2\eta_t \nabla \ell_t(\boldsymbol{\theta}_t)^\top (\boldsymbol{\theta}_t - \boldsymbol{\theta}^\star)
+ \eta_t^2 \|\nabla \ell_t(\boldsymbol{\theta}_t)\|^2.
\]

Rearranging,
\[
\nabla \ell_t(\boldsymbol{\theta}_t)^\top (\boldsymbol{\theta}_t - \boldsymbol{\theta}^\star)
\le
\frac{\|\boldsymbol{\theta}_t - \boldsymbol{\theta}^\star\|^2 - \|\boldsymbol{\theta}_{t+1} - \boldsymbol{\theta}^\star\|^2}{2\eta_t}
+ \frac{\eta_t G^2}{2}.
\]

\medskip
\noindent
\textbf{Step 2: Convexity.}

By convexity of $\ell_t$,
\[
\ell_t(\boldsymbol{\theta}_t) - \ell_t(\boldsymbol{\theta}^\star)
\le
\nabla \ell_t(\boldsymbol{\theta}_t)^\top (\boldsymbol{\theta}_t - \boldsymbol{\theta}^\star).
\]

Combining,
\[
\ell_t(\boldsymbol{\theta}_t) - \ell_t(\boldsymbol{\theta}^\star)
\le
\frac{\|\boldsymbol{\theta}_t - \boldsymbol{\theta}^\star\|^2 - \|\boldsymbol{\theta}_{t+1} - \boldsymbol{\theta}^\star\|^2}{2\eta_t}
+ \frac{\eta_t G^2}{2}.
\]

\medskip
\noindent
\textbf{Step 3: Summation.}

Summing over $t=1,\dots,T$,
\[
\sum_{t=1}^T \bigl(\ell_t(\boldsymbol{\theta}_t) - \ell_t(\boldsymbol{\theta}^\star)\bigr)
\le
\frac{D^2}{2\eta_T} + \frac{G^2}{2}\sum_{t=1}^T \eta_t.
\]

\medskip
\noindent
\textbf{Step 4: Step-size choice.}

Using a horizon-free schedule $\eta_t = \eta_0 / \sqrt{t}$,
\[
\sum_{t=1}^T \eta_t \le 2\eta_0 \sqrt{T}.
\]
Thus,
\[
\text{Regret}_T^{\text{static}}
= \sum_{t=1}^T \left[\ell_t(\boldsymbol{\theta}_t) - \ell_t(\boldsymbol{\theta}^\star)\right]
\le
O(\sqrt{T}).
\]

\end{proof}

\subsection{Proof of Theorem~5 (Dynamic regret)}
\label{app:proof:dynamic-regret}

\begin{proof}
Let $\{\boldsymbol{\theta}_t\}$ be generated by the projected online-gradient update used in the proof of Theorem~4.

Define the comparator sequence $\{\boldsymbol{\theta}_t^\star\}$ and dynamic regret:
\[
\text{Regret}_T^{\text{dyn}}
=
\sum_{t=1}^T \left[\ell_t(\boldsymbol{\theta}_t) - \ell_t(\boldsymbol{\theta}_t^\star)\right].
\]

Assume drift is bounded:
\[
V_T := \sum_{t=2}^T \|\boldsymbol{\theta}_t^\star - \boldsymbol{\theta}_{t-1}^\star\|.
\]

\medskip
\noindent
\textbf{Step 1: Decomposition.}

Using the same one-step expansion as in static regret but with time-varying comparator:
\[
\ell_t(\boldsymbol{\theta}_t) - \ell_t(\boldsymbol{\theta}_t^\star)
\le
\frac{\|\boldsymbol{\theta}_t - \boldsymbol{\theta}_t^\star\|^2 - \|\boldsymbol{\theta}_{t+1} - \boldsymbol{\theta}_t^\star\|^2}{2\eta_t}
+ \frac{\eta_t G^2}{2}.
\]

\medskip
\noindent
\textbf{Step 2: Drift handling.}

Insert and subtract $\boldsymbol{\theta}_{t+1}^\star$:
\[
\|\boldsymbol{\theta}_{t+1} - \boldsymbol{\theta}_t^\star\|^2
=
\|\boldsymbol{\theta}_{t+1} - \boldsymbol{\theta}_{t+1}^\star + \boldsymbol{\theta}_{t+1}^\star - \boldsymbol{\theta}_t^\star\|^2.
\]

Using triangle inequality and standard bounds,
\[
\|\boldsymbol{\theta}_{t+1} - \boldsymbol{\theta}_t^\star\|^2
\ge
\|\boldsymbol{\theta}_{t+1} - \boldsymbol{\theta}_{t+1}^\star\|^2
- 2D \|\boldsymbol{\theta}_{t+1}^\star - \boldsymbol{\theta}_t^\star\|.
\]

\medskip
\noindent
\textbf{Step 3: Summation.}

Summing over $t$ and telescoping,
\[
\text{Regret}_T^{\text{dyn}}
\le
O\Big(\frac{D^2}{\eta_T} + G^2 \sum_{t=1}^T \eta_t + \frac{D}{\eta_T} V_T\Big).
\]

\medskip
\noindent
\textbf{Step 4: Step-size choice.}

With $\eta_t = \eta_0 / \sqrt{t}$,
\[
\text{Regret}_T^{\text{dyn}}
\le
O(\sqrt{T} + V_T \sqrt{T}).
\]

\end{proof}
% \textbf{Remark (Why OCO is appropriate).}
Although the problem is framed as parameter recovery rather than decision-making, 
the online estimator defines a sequential optimization problem under evolving 
inverse losses. The regret guarantees quantify how well the online updates track 
the sequence of minimizers of the inverse loss, thereby providing performance 
guarantees for real-time parameter estimation under drift.
% \textbf{Remark (Convexity of the inverse loss).}
The regret analysis assumes convexity of $\ell_t(\boldsymbol{\theta})$. 
In the linear-gradient setting, the projected-gradient term is quadratic in 
$\boldsymbol{\theta}$, which yields convexity. More generally, the analysis 
applies to convex approximations or local convexifications of the inverse loss, 
which are standard in inverse optimization and online estimation.
% \textbf{Remark (Tightness of regret bounds).}
The obtained bounds have the standard online-convex-optimization form under the bounded-gradient and comparator-variation assumptions stated in the main text.

\subsection{Proof of Theorem~6 (Residual-to-parameter stability)}
\label{app:proof:noise}

\begin{proof}
The proof analyzes stability of parameter recovery under perturbations in the inverse loss.

\noindent
\textbf{Step 1: Residual definition.}

Define the projected-gradient residual
\[
\mathbf{r}_t(\boldsymbol{\theta})
:=
\Proj_t \nabla_{\mathbf{x}} c_t(\mathbf{x}_t;\boldsymbol{\theta}).
\]

Under the linear-gradient structure,
\[
\mathbf{r}_t(\boldsymbol{\theta})
=
\widetilde{\mathbf{A}}_t \boldsymbol{\theta} + \widetilde{\mathbf{b}}_t,
\quad
\widetilde{\mathbf{A}}_t := \Proj_t \mathbf{A}_t(\mathbf{x}_t).
\]

For the true parameter $\boldsymbol{\theta}_t$, define
\[
\mathbf{r}_t(\boldsymbol{\theta}_t)
:=
\boldsymbol{\epsilon}_t,
\qquad
\|\boldsymbol{\epsilon}_t\| \le \delta_t,
\]
where $\boldsymbol{\epsilon}_t$ captures model misspecification, noise, or approximate KKT consistency.

\noindent
\textbf{Step 2: Difference relation.}

For any $\boldsymbol{\theta}$,
\[
\mathbf{r}_t(\boldsymbol{\theta}) - \mathbf{r}_t(\boldsymbol{\theta}_t)
=
\widetilde{\mathbf{A}}_t(\boldsymbol{\theta} - \boldsymbol{\theta}_t),
\]
and hence
\[
\widetilde{\mathbf{A}}_t(\boldsymbol{\theta} - \boldsymbol{\theta}_t)
=
\mathbf{r}_t(\boldsymbol{\theta}) - \boldsymbol{\epsilon}_t.
\]

\noindent
\textbf{Step 3: Stability bound.}

If $\widetilde{\mathbf{A}}_t$ has full column rank, then
\[
\sigma_{\min}(\widetilde{\mathbf{A}}_t)\,
\|\boldsymbol{\theta} - \boldsymbol{\theta}_t\|
\le
\|\mathbf{r}_t(\boldsymbol{\theta}) - \boldsymbol{\epsilon}_t\|.
\]

\noindent
\textbf{Step 4: Error decomposition.}

By the triangle inequality,
\[
\|\mathbf{r}_t(\boldsymbol{\theta}) - \boldsymbol{\epsilon}_t\|
\le
\|\mathbf{r}_t(\boldsymbol{\theta})\| + \delta_t.
\]

Therefore,
\[
\|\boldsymbol{\theta} - \boldsymbol{\theta}_t\|
\le
\frac{1}{\sigma_{\min}(\widetilde{\mathbf{A}}_t)}
\bigl(
\|\mathbf{r}_t(\boldsymbol{\theta})\| + \delta_t
\bigr).
\]

\noindent
\textbf{Step 5: Relation to inverse loss.}

For the projected online loss used by Algorithm~1,
\[
\ell_t(\boldsymbol{\theta}) = \|\mathbf{r}_t(\boldsymbol{\theta})\|^2.
\]

Therefore,
\[
\|\boldsymbol{\theta} - \boldsymbol{\theta}_t\|
\le
\frac{1}{\sigma_{\min}(\widetilde{\mathbf{A}}_t)}
\Bigl(
\sqrt{\ell_t(\boldsymbol{\theta})} + \delta_t
\Bigr),
\]
which proves the claimed deterministic stability bound.
\end{proof}

% \noindent
% \textbf{Interpretation.}
Accurate parameter recovery does not require exact consistency. 
Small inverse loss directly implies small parameter error, up to a perturbation term 
controlled by $\delta_t$.
% \noindent
% \textbf{Remark (Error decomposition).}
The bound separates two sources of estimation error:
(i) optimization error captured by $\ell_t(\boldsymbol{\theta})$, and 
(ii) data or model noise captured by $\delta_t$. 
Reliable recovery is achieved when both terms are controlled.
% \noindent
% \textbf{Remark (Conditioning).}
The factor $\sigma_{\min}(\widetilde{\mathbf{A}}_t)$ reflects the conditioning of the 
inverse problem. When it is small, the problem becomes ill-conditioned, and parameter 
estimates become sensitive to perturbations, indicating limited information in the data.

\clearpage
\section{Additional Computational Specifications and Diagnostics}
\label{app:computational}

\subsection{Detailed Benchmark Configurations}
The following tables provide the full benchmark specifications underlying the compact summary in the main manuscript. The settings are synthetic and are intended for controlled diagnostics rather than field calibration. For consistency with the main manuscript, the two settings correspond to the Multi-resource and Network-dispatch benchmarks. Compact table and plot labels inherited from the computational outputs use \emph{Healthcare} and \emph{Energy} as shorthand for these two stylized benchmarks, respectively; those labels provide operational motivation and do not imply field calibration.
For standalone reading, confidence interval (CI), mean squared error (MSE), and online convex optimization (OCO) are written in full at first use in this supplement.

\begin{sidewaystable}[!htbp]
\centering
\caption{Theory-aligned experimental design for the Multi-resource benchmark (Part I: Structural model and inverse formulation). 
This setting emphasizes short-horizon responsiveness, identifiability-preserving inverse recovery, 
and high-precision tracking under controlled drift.}
\label{tab:healthcare-final-part1}
\scriptsize
\renewcommand{\arraystretch}{0.92}
\setlength{\tabcolsep}{3pt}

\begin{tabular}{p{2.8cm}p{3.2cm}p{5.1cm}p{4.3cm}p{5.3cm}}
\toprule
Component 
& Specification 
& Mathematical Form 
& Theory Link 
& Operational Meaning \\
\midrule

System size 
& $n=5$, $k=2$, $p=3$, $T\in\{200,400,800,1600\}$ 
& $x_t\in\mathbb{R}_+^5,\;\theta_t\in\Theta\subset\mathbb{R}^3$ 
& Dynamic regret; identifiability 
& Five allocation components, two coupled resource constraints, short-to-medium horizon \\

Feasible set 
& Capacity-constrained multi-resource allocation 
& $X_t=\{x\in\mathbb{R}_+^5: Bx\le q\}$ 
& KKT structure; well-posedness 
& Allocation subject to two coupled capacity/resource limits \\

Constraint matrix 
& Fixed identifiable structure 
& $B\in\mathbb{R}^{2\times 5}$ with full-row-rank and projector $P=I-B^\top(BB^\top)^\dagger B$ 
& Lemma~1; Theorem~2 
& Stable operational rules and projected-gradient identification \\

Latent preference path 
& Controlled drift with bounded variation 
& $\theta_t=\theta_{t-1}+\delta_t,\;\|\delta_t\|=O(1/\sqrt{T})$ 
& Theorem~5 
& Slowly evolving latent priority weights in the multi-resource interpretation \\

Variation budget 
& Bounded cumulative variation 
& $V_T=\sum_{t=2}^T \|\theta_t-\theta_{t-1}\| = O(\sqrt{T})$ 
& Theorem~5 
& Nonstationarity remains controlled over the horizon \\

Forward cost function 
& Fairness-aware strongly convex cost 
& $c_t(x;\theta_t)=\frac{\rho_h}{2}\|x\|^2-\beta_h(A\theta_t)^\top x+\lambda_f\|x-\bar{x}\|^2$ 
& Assumptions~3--4; Theorem~3 
& Trade-off among convex allocation cost, latent priorities, and deviation from a reference allocation \\

Forward decision generation 
& Synthetic optimal allocation path 
& $x_t^\star=\arg\min_{x\in X_t} c_t(x;\theta_t)$ 
& Forward--inverse consistency; Assumption~1 
& Ground-truth decisions induced by the latent preference path \\

Observed decision model 
& Noisy but theory-aligned observations 
& $\tilde{x}_t=x_t^\star+\varepsilon_t,\;\varepsilon_t\sim \mathcal{N}(0,\sigma^2 I)$, followed by projection onto $X_t$ if needed 
& Theorem~6 
& Perturbed observations in the multi-resource benchmark \\

Noise levels 
& Low-to-moderate uncertainty 
& $\sigma^2\in\{0,0.01,0.05\}$ 
& Perturbation stability 
& Clean, mild-noise, and moderate-noise benchmark settings \\

Identifiability regime 
& Projected Jacobian full rank 
& $\tilde{A}_t=P A_t(\tilde{x}_t)$ has full column rank 
& Theorem~2; Corollary~2 
& Exact recovery possible in the noiseless identifiable regime \\

Inverse residual 
& Projected KKT residual 
& $r_t(\theta)=P\nabla_x c_t(\tilde{x}_t;\theta)$ 
& Lemma~3; Algorithm~1 
& Measures structural mismatch between a candidate preference vector and observed decisions \\

Inverse loss 
& Quadratic inverse loss 
& $\ell_t(\theta)=\|r_t(\theta)\|^2$ 
& Theorems~4--6 
& Per-period structural inconsistency score \\

Parameter domain 
& Compact projected parameter space 
& $\Theta=[0,5]^3$ 
& Assumption~5 
& Bounded latent priority weights \\

Estimator update 
& Online projected gradient / mirror-descent implementation 
& $\hat{\theta}_{t+1}=\Pi_{\Theta}\bigl(\hat{\theta}_t-\eta_t g_t\bigr),\;\; g_t\in\partial \ell_t(\hat{\theta}_t)$ 
& Algorithm~1; Theorems~4--5 
& Sequential updating of inferred preference weights \\

Step size 
& Aggressive adaptation 
& $\eta_t\propto t^{-1/2}$ (domain-specific constant) 
& Convergence; regret control 
& Faster adaptation under the recorded benchmark setting \\

Initialization 
& Common cold-start estimator 
& $\hat{\theta}_1=\theta_{\mathrm{init}}$ with fixed initialization across runs 
& Reproducibility 
& Standardized initial ignorance before learning begins \\

\bottomrule
\end{tabular}
\end{sidewaystable}

\begin{sidewaystable}[!htbp]
\centering
\caption{Theory-aligned experimental design for the Multi-resource benchmark (Part II: Evaluation and diagnostics). 
This setting emphasizes short-horizon responsiveness, identifiability-preserving inverse recovery, 
and high-precision tracking under controlled drift.}
\label{tab:healthcare-final-part2}
\scriptsize
\renewcommand{\arraystretch}{0.92}
\setlength{\tabcolsep}{3pt}

\begin{tabular}{p{2.8cm}p{3.2cm}p{5.1cm}p{4.3cm}p{5.3cm}}
\toprule
Component 
& Specification 
& Mathematical Form 
& Theory Link 
& Operational Meaning \\
\midrule

Static comparator 
& Best fixed parameter in hindsight 
& $\theta^\star\in\arg\min_{\theta\in\Theta}\sum_{t=1}^T \ell_t(\theta)$ 
& Theorem~4 
& Benchmark against a time-invariant policy explanation \\

Dynamic comparator 
& True time-varying target path 
& $\{\theta_t\}_{t=1}^T$ 
& Theorem~5 
& Benchmark against the evolving latent preference path \\

Static regret 
& Fixed-comparator cumulative regret 
& $R_T^{\mathrm{stat}}=\sum_{t=1}^T \ell_t(\hat{\theta}_t)-\min_{\theta\in\Theta}\sum_{t=1}^T \ell_t(\theta)$ 
& Theorem~4 
& Loss relative to the best single static explanation \\

Dynamic regret 
& Pathwise cumulative regret 
& $R_T^{\mathrm{dyn}}=\sum_{t=1}^T \ell_t(\hat{\theta}_t)-\sum_{t=1}^T \ell_t(\theta_t)$ 
& Theorem~5 
& Loss relative to the true evolving preference path \\

Tracking error 
& Instantaneous estimation error 
& $\|\hat{\theta}_t-\theta_t\|$ 
& Theorem~6; empirical diagnostic 
& Pointwise recovery accuracy for the latent parameter path \\

Diagnostic outputs 
& Scaling and robustness diagnostics 
& Log--log slope of regret, final tracking error, normalized regret $R_T^{\mathrm{dyn}}/[\sqrt{T}(1+V_T)]$ 
& Theorems~4--6 
& Finite-sample assessment of scaling patterns \\

Monte Carlo design 
& Repeated runs with uncertainty bands 
& Repeated independent runs; mean $\pm$ 95\% CI where applicable 
& Empirical reliability 
& Stable camera-ready summaries and interval reporting \\

\midrule
\multicolumn{5}{l}{Diagnostic targets} \\

Identifiability 
& Noiseless exact recovery check 
& $\sigma^2=0$ and full-rank $\tilde{A}_t$ 
& Theorem~2 
& Structural correctness of inverse recovery \\

Regret scaling 
& Horizon growth under bounded variation 
& $R_T^{\mathrm{stat}}$ and $R_T^{\mathrm{dyn}}$ versus $T$ on log--log axes 
& Theorems~4--5 
& Whether cumulative estimation error remains sublinear \\

Tracking precision 
& Final and pathwise recovery error 
& $\|\hat{\theta}_t-\theta_t\|$ and $\|\hat{\theta}_T-\theta_T\|$ 
& Theorem~6 
& Accuracy of latent-parameter recovery \\

Convergence speed 
& Early stabilization 
& Relative decline of tracking error over time 
& Dynamic tracking 
& Rapid adaptation after initialization and under mild drift \\

Perturbation stability 
& Error sensitivity to observation noise 
& Tracking error and regret as functions of $\sigma^2$ 
& Theorem~6 
& Sensitivity under perturbed observations \\

\bottomrule
\end{tabular}
\end{sidewaystable}

\begin{sidewaystable}[!htbp]
\centering
\caption{Theory-aligned experimental design for the Network-dispatch benchmark (Part I: Structural model and inverse formulation). 
This setting emphasizes long-horizon accumulation effects, persistent drift, 
and stability of inverse recovery under stronger nonstationarity.}
\label{tab:energy-final-part1}
\scriptsize
\renewcommand{\arraystretch}{0.92}
\setlength{\tabcolsep}{3pt}

\begin{tabular}{p{2.8cm}p{3.2cm}p{5.1cm}p{4.3cm}p{5.3cm}}
\toprule
Component 
& Specification 
& Mathematical Form 
& Theory Link 
& Operational Meaning \\
\midrule

System size 
& $n=4$, $k=1$, $p=3$, $T\in\{400,800,1600,2000\}$ 
& $x_t\in\mathbb{R}_+^4,\;\theta_t\in\Theta\subset\mathbb{R}^3$ 
& Dynamic regret; drift scaling 
& Four dispatch components, one aggregate resource constraint, long horizon \\

Feasible set 
& Network-constrained dispatch 
& $X_t=\{x\in\mathbb{R}_+^4: Bx\le q\}$ 
& KKT structure; well-posedness 
& Stylized dispatch under an aggregate capacity/resource constraint \\

Constraint matrix 
& Time-invariant network structure 
& $B\in\mathbb{R}^{1\times 4}$ with projector $P=I-B^\top(BB^\top)^\dagger B$ 
& Lemma~1; Theorem~2 
& Fixed aggregate resource relation in the dispatch benchmark \\

Latent preference path 
& Persistent drift 
& $\theta_t=\theta_{t-1}+\delta_t,\;\|\delta_t\|=O(1/\sqrt{T})$ 
& Theorem~5 
& Slowly evolving latent dispatch-preference weights \\

Variation budget 
& Accumulating variation 
& $V_T=O(\sqrt{T})$ 
& Theorem~5 
& Long-horizon drift accumulation \\

Forward cost function 
& Convex economic dispatch objective 
& $c_t(x;\theta_t)=\frac{\rho_e}{2}\|x\|^2-\beta_e(A\theta_t)^\top x+\gamma e^\top x$ 
& Assumptions~3--4; Theorem~3 
& Trade-off among convex dispatch cost, latent preferences, and the linear dispatch penalty \\

Forward decision generation 
& Optimal dispatch trajectory 
& $x_t^\star=\arg\min_{x\in X_t} c_t(x;\theta_t)$ 
& Forward--inverse consistency 
& Ground-truth dispatch decisions \\

Observed decision model 
& Noisy dispatch observations 
& $\tilde{x}_t=x_t^\star+\varepsilon_t,\;\varepsilon_t\sim \mathcal{N}(0,\sigma^2 I)$ 
& Theorem~6 
& Perturbed observations in the stylized dispatch benchmark \\

Noise levels 
& Moderate-to-high uncertainty 
& $\sigma^2\in\{0,0.05,0.10\}$ 
& Perturbation stability 
& Tested moderate-to-higher noise settings \\

Identifiability regime 
& Full-rank projected Jacobian 
& $\tilde{A}_t=P A_t(\tilde{x}_t)$ full rank 
& Theorem~2 
& Recoverability under network constraints \\

Inverse residual 
& Projected gradient residual 
& $r_t(\theta)=P\nabla_x c_t(\tilde{x}_t;\theta)$ 
& Lemma~3 
& Structural mismatch under dispatch constraints \\

Inverse loss 
& Quadratic loss 
& $\ell_t(\theta)=\|r_t(\theta)\|^2$ 
& Theorems~4--6 
& Per-period inverse consistency \\

Parameter domain 
& Bounded parameter set 
& $\Theta=[0,2]^3$ 
& Assumption~5 
& Bounded latent preference coefficients \\

Estimator update 
& Online mirror descent 
& $\hat{\theta}_{t+1}=\Pi_{\Theta}(\hat{\theta}_t-\eta_t g_t)$ 
& Algorithm~1 
& Sequential learning under drift \\

Step size 
& Conservative adaptation 
& $\eta_t\propto t^{-1/2}$ (domain-specific constant) 
& Regret control 
& Stable updates under stronger drift \\

Initialization 
& Common initialization 
& $\hat{\theta}_1=\theta_{\mathrm{init}}$ 
& Reproducibility 
& Standard initial condition \\

\bottomrule
\end{tabular}
\end{sidewaystable}

\begin{sidewaystable}[!htbp]
\centering
\caption{Theory-aligned experimental design for the Network-dispatch benchmark (Part II: Evaluation and diagnostics). 
This setting emphasizes long-horizon accumulation effects, sensitivity to persistent drift, 
and the behavior of trajectory-aware estimation under noise and temporal instability.}
\label{tab:energy-final-part2}
\scriptsize
\renewcommand{\arraystretch}{0.92}
\setlength{\tabcolsep}{3pt}

\begin{tabular}{p{2.8cm}p{3.2cm}p{5.1cm}p{4.3cm}p{5.3cm}}
\toprule
Component 
& Specification 
& Mathematical Form 
& Theory Link 
& Operational Meaning \\
\midrule

Static comparator 
& Best fixed parameter in hindsight 
& $\theta^\star\in\arg\min_{\theta\in\Theta}\sum_{t=1}^T \ell_t(\theta)$ 
& Theorem~4 
& Benchmark against a time-invariant dispatch policy \\

Dynamic comparator 
& True time-varying trajectory 
& $\{\theta_t\}_{t=1}^T$ 
& Theorem~5 
& Benchmark against evolving cost and demand conditions \\

Static regret 
& Fixed-comparator cumulative regret 
& $R_T^{\mathrm{stat}}=\sum_{t=1}^T \ell_t(\hat{\theta}_t)-\min_{\theta\in\Theta}\sum_{t=1}^T \ell_t(\theta)$ 
& Theorem~4 
& Loss relative to the best static explanation of dispatch decisions \\

Dynamic regret 
& Pathwise cumulative regret 
& $R_T^{\mathrm{dyn}}=\sum_{t=1}^T \ell_t(\hat{\theta}_t)-\sum_{t=1}^T \ell_t(\theta_t)$ 
& Theorem~5 
& Loss relative to the true evolving system dynamics \\

Tracking error 
& Instantaneous estimation error 
& $\|\hat{\theta}_t-\theta_t\|$ 
& Theorem~6; empirical diagnostic 
& Accuracy of recovering latent dispatch preference parameters \\

Diagnostic outputs 
& Scaling and robustness diagnostics 
& Log--log slope of regret, final tracking error, normalized regret $R_T^{\mathrm{dyn}}/[\sqrt{T}(1+V_T)]$ 
& Theorems~4--6 
& Finite-sample comparison with the drift-sensitive scaling envelope \\

Monte Carlo design 
& Repeated stochastic experiments 
& Repeated independent runs; mean $\pm$ 95\% CI where applicable 
& Empirical reliability 
& Stability of results under stochastic dispatch fluctuations \\

\midrule
\multicolumn{5}{l}{Diagnostic targets} \\

Identifiability under drift 
& Near-identifiable but noisy regime 
& $\sigma^2 > 0$, full-rank $\tilde{A}_t$ 
& Theorem~2 
& Recoverability under imperfect and evolving system observations \\

Regret scaling under drift 
& Long-horizon accumulation 
& $R_T^{\mathrm{stat}}, R_T^{\mathrm{dyn}}$ vs $T$ (log--log) 
& Theorems~4--5 
& Detection of deviation from $\sqrt{T}$ scaling under drift \\

Tracking degradation 
& Error accumulation over time 
& $\|\hat{\theta}_t-\theta_t\|$ growth patterns 
& Theorem~6 
& Sensitivity of inverse recovery to persistent system changes \\

Drift sensitivity 
& Strong variation regime 
& $V_T=O(\sqrt{T})$ 
& Theorem~5 
& Stress-testing the limits of dynamic regret guarantees \\

Scaling deviation 
& Departure from ideal slope 
& log--log slope $>0.5$ (empirical) 
& Theory-related diagnostic 
& Identification of drift-dominated behavior \\

Perturbation stability 
& High-noise evaluation 
& Tracking error and regret vs $\sigma^2$ 
& Theorem~6 
& Stability under measurement noise and forecast uncertainty \\

Online-estimation comparison 
& Online vs naive comparison 
& $R_T^{\mathrm{dyn}}(\mathrm{naive}) / R_T^{\mathrm{dyn}}(\mathrm{OCO})$ 
& Algorithm 1; Theorem 5 
& Compares stability under alternative update rules \\

\bottomrule
\end{tabular}
\end{sidewaystable}

\clearpage
\subsection{Additional Regret and Online-Update Diagnostics}
\begin{table}[!htpb]
\centering
\caption{Effect of noise variance on tracking error and dynamic regret.}
\label{tab:noise_robustness}
\scriptsize
\begin{tabular}{l c c c c}
\toprule
Domain & $T$ & $\sigma^2$ & Tracking Error & Dynamic Regret \\
\midrule
Healthcare & 800  & 0.00 & 1.126 & 148.27 \\
Healthcare & 800  & 0.01 & 1.151 & 178.04 \\
Healthcare & 800  & 0.05 & 1.182 & 281.21 \\
\midrule
Energy     & 1600 & 0.00 & 1.547 & 278.48 \\
Energy     & 1600 & 0.05 & 1.570 & 353.15 \\
Energy     & 1600 & 0.10 & 1.577 & 419.90 \\
\bottomrule
\end{tabular}
\end{table}

\begin{figure}[!htpb]
\centering
\includegraphics[width=\linewidth]{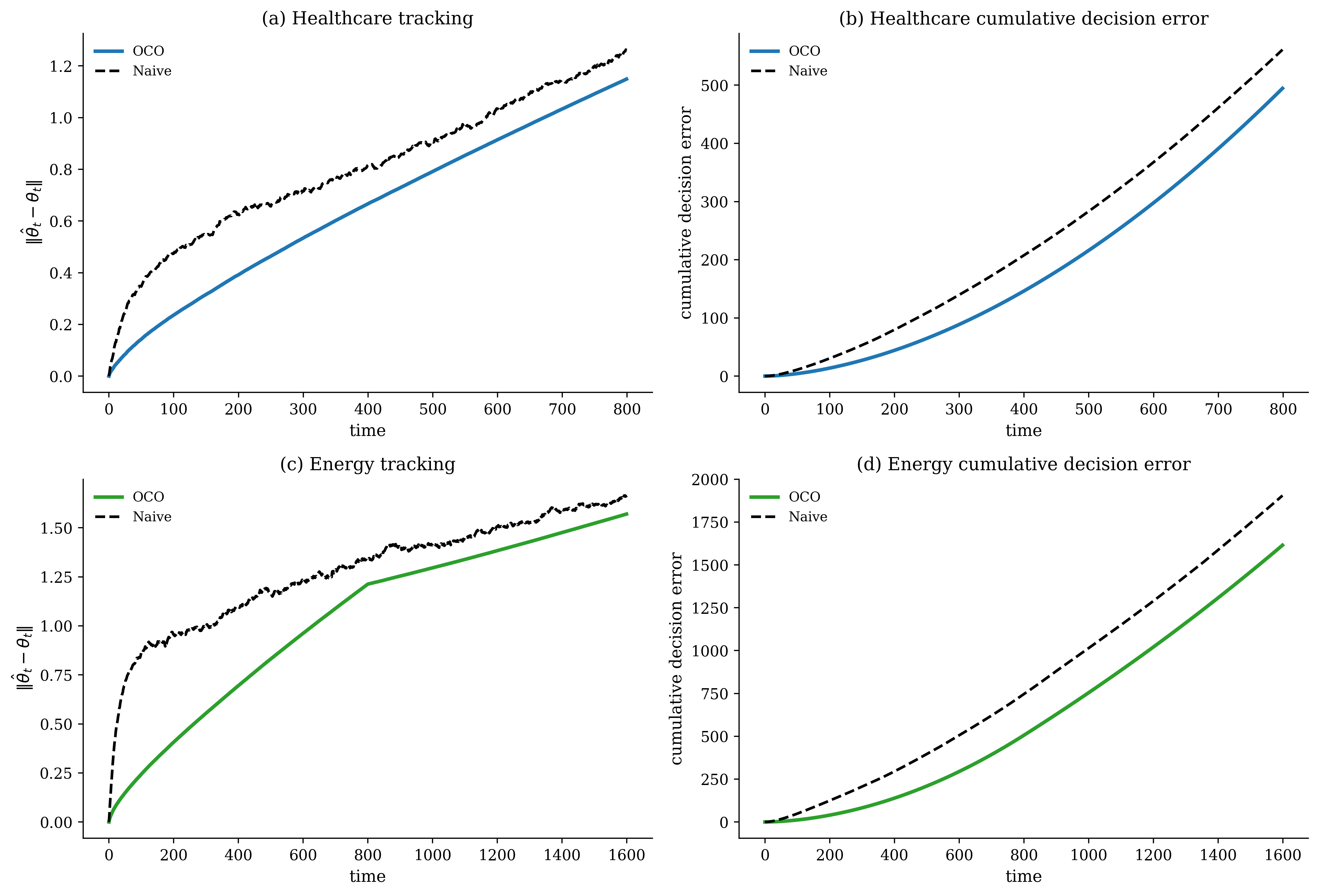}
\caption{
Comparison between OCO and naive estimators.
(a,c) Tracking error over time.
(b,d) cumulative forward decision error.
}
\label{fig:oco_necessity}
\end{figure}

\clearpage
\subsection{Detailed Noise Diagnostics}
\begin{table}[!htpb]
\centering
\caption{Noise sensitivity under stationary conditions}
\label{tab:noise-theorem}
\scriptsize
\begin{tabular}{lcccccc}
\toprule
Domain & $\sigma$ & $\sigma^2$ & Late MSE & CI (95\%) & Tracking Error \\
\midrule
Healthcare & 0.00 & 0.0000 & 0.257 & 0.170 & 0.510 \\
Healthcare & 0.05 & 0.0025 & 0.405 & 0.220 & 0.602 \\
Healthcare & 0.10 & 0.0100 & 0.308 & 0.136 & 0.528 \\
Healthcare & 0.30 & 0.0900 & 0.262 & 0.117 & 0.522 \\
\midrule
Energy     & 0.00 & 0.0000 & 1.654 & 0.380 & 1.253 \\
Energy     & 0.05 & 0.0025 & 1.835 & 0.294 & 1.369 \\
Energy     & 0.10 & 0.0100 & 1.418 & 0.364 & 1.176 \\
Energy     & 0.30 & 0.0900 & 1.460 & 0.360 & 1.183 \\
\bottomrule
\end{tabular}
\end{table}

\begin{table}[!htpb]
\centering
\caption{Projection-induced noise attenuation}
\label{tab:projection-diagnostic}
\scriptsize
\begin{tabular}{lcccc}
\toprule
Domain & $\sigma$ & Raw Residual & Projected Residual & Attenuation Ratio \\
\midrule
Healthcare & 0.02 & 0.245 & 0.126 & 0.591 \\
Healthcare & 0.10 & 0.098 & 0.059 & 0.604 \\
Healthcare & 0.30 & 0.890 & 0.534 & 0.599 \\
\midrule
Energy     & 0.02 & 6.781 & 3.066 & 0.487 \\
Energy     & 0.10 & 6.449 & 3.084 & 0.578 \\
Energy     & 0.30 & 9.131 & 4.291 & 0.492 \\
\bottomrule
\end{tabular}
\end{table}

\begin{table}[!htpb]
\centering
\caption{Performance under combined drift and noise}
\label{tab:noise-drift}
\scriptsize
\begin{tabular}{lcccccc}
\toprule
Domain & Drift & $\sigma$ & Tracking Error & CI (95\%) & Inv. Regret & Dec. Regret \\
\midrule
Healthcare & 0.5 & 0.00 & 0.818 & 0.184 & 150.53 & 522.70 \\
Healthcare & 1.0 & 0.00 & 1.309 & 0.202 & 283.92 & 887.15 \\
Healthcare & 2.0 & 0.00 & 1.915 & 0.306 & 612.38 & 3139.59 \\
\midrule
Energy     & 0.5 & 0.00 & 1.590 & 0.170 & 15.67  & 963.98 \\
Energy     & 1.0 & 0.00 & 2.028 & 0.153 & 30.88  & 1081.77 \\
Energy     & 2.0 & 0.00 & 2.765 & 0.289 & 480.78 & 5253.10 \\
\bottomrule
\end{tabular}
\end{table}

\clearpage
\subsection{Extended Shock Diagnostics}
The cumulative post-shock burden is summarized by the area under the curve (AUC).
\begin{table*}[!htbp]
\centering
\caption{Shock response decomposition with statistical comparison across the two stylized benchmarks. All values are means with 95\% confidence intervals. $p$-values correspond to two-sample $t$-tests comparing the Multi-resource (Healthcare shorthand) and Network-dispatch (Energy shorthand) benchmarks at the same shock magnitude.}
\label{tab:shock_rigorous}
\scriptsize
\setlength{\tabcolsep}{3pt}
\renewcommand{\arraystretch}{1.1}

\begin{tabular}{llccccccc}
\toprule
\multicolumn{9}{c}{Panel A: Baseline and Persistence} \\
\midrule
Domain & Shock & Baseline & CI & Late Error & CI & Residual Ratio & CI & p-value \\
\midrule
Healthcare & 0.5 & 0.881 & 0.160 & 0.949 & 0.184 & 1.132 & 0.433 & \multirow{2}{*}{$<0.01$} \\
Energy     & 0.5 & 1.431 & 0.163 & 1.559 & 0.180 & 1.877 & 0.136 &  \\
\midrule
Healthcare & 1.0 & 0.809 & 0.131 & 0.915 & 0.151 & 1.125 & 0.221 & \multirow{2}{*}{$<0.01$} \\
Energy     & 1.0 & 1.582 & 0.182 & 1.777 & 0.206 & 1.602 & 0.127 &  \\
\midrule
Healthcare & 2.0 & 0.840 & 0.157 & 1.009 & 0.188 & 0.932 & 0.145 & \multirow{2}{*}{$<0.01$} \\
Energy     & 2.0 & 1.527 & 0.151 & 1.789 & 0.175 & 1.277 & 0.082 &  \\
\midrule

\multicolumn{9}{c}{Panel B: Transient Amplification} \\
\midrule
Domain & Shock & Peak Error & CI & Peak Uplift & CI & Normalized Peak & CI & p-value \\
\midrule
Healthcare & 0.5 & 0.930 & 0.167 & 0.049 & 0.007 & 1.064 & 0.014 & \multirow{2}{*}{$<0.01$} \\
Energy     & 0.5 & 1.498 & 0.168 & 0.067 & 0.006 & 1.048 & 0.004 &  \\
\midrule
Healthcare & 1.0 & 0.897 & 0.137 & 0.088 & 0.009 & 1.118 & 0.015 & \multirow{2}{*}{$<0.01$} \\
Energy     & 1.0 & 1.701 & 0.189 & 0.119 & 0.008 & 1.079 & 0.008 &  \\
\midrule
Healthcare & 2.0 & 1.018 & 0.165 & 0.178 & 0.012 & 1.300 & 0.138 & \multirow{2}{*}{$<0.01$} \\
Energy     & 2.0 & 1.730 & 0.158 & 0.204 & 0.008 & 1.138 & 0.010 &  \\
\midrule

\multicolumn{9}{c}{Panel C: Recovery and Burden} \\
\midrule
Domain & Shock & Residual & CI & Recovery Time & CI & Burden (AUC) & CI & p-value \\
\midrule
Healthcare & 0.5 & 0.068 & 0.027 & 0 & 0.0 & 20.94 & 7.00 & \multirow{2}{*}{$<0.01$} \\
Energy     & 0.5 & 0.128 & 0.019 & 0 & 0.0 & 36.73 & 5.61 &  \\
\midrule
Healthcare & 1.0 & 0.106 & 0.028 & 0 & 0.0 & 27.96 & 7.71 & \multirow{2}{*}{$<0.01$} \\
Energy     & 1.0 & 0.195 & 0.025 & 0 & 0.0 & 52.35 & 7.12 &  \\
\midrule
Healthcare & 2.0 & 0.170 & 0.034 & 145.9 & 124.7 & 39.56 & 9.04 & \multirow{2}{*}{$<0.01$} \\
Energy     & 2.0 & 0.263 & 0.025 & $\infty$ & -- & 64.17 & 7.14 &  \\
\bottomrule
\end{tabular}
\end{table*}